\documentclass[%
 reprint,
 amsmath,amssymb,
 aps,prx,
 longbibliography,
 superscriptaddress
]{revtex4-1}

\usepackage{standalone}
\usepackage{graphicx}
\usepackage{bbold}

\usepackage{bm}
\usepackage{braket}
\usepackage{float}
\usepackage{enumitem}

\newcommand{\innerproduct}[2]{\left<#1|#2\right>}
\newcommand{\outerproduct}[2]{\left|#1\right>\left<#2\right|}
\newcommand{\Tr}{\mathrm{Tr}}

\usepackage[T1]{fontenc}
\usepackage[latin9]{inputenc}

\usepackage{amsthm}
\usepackage{amsmath}
\usepackage{amssymb}

\makeatletter
\theoremstyle{plain}
\newtheorem{thm}{\protect\theoremname}
  \theoremstyle{definition}
  \newtheorem{defn}[thm]{\protect\definitionname}
  \theoremstyle{plain}
  \newtheorem{lem}[thm]{\protect\lemmaname}
\makeatother

\usepackage[english]{babel}
  \providecommand{\definitionname}{Definition}
  \providecommand{\lemmaname}{Lemma}
\providecommand{\theoremname}{Theorem}
\usepackage{hyperref}

\begin{document}
\preprint{APS/123-QED}

\title{Implementing and characterizing precise multi-qubit measurements}

\author{J. Z. Blumoff}
\thanks{These two authors contributed equally.  Corresponding author:  kevin.chou@yale.edu}
\affiliation{Department of Applied Physics and Physics, Yale University, New Haven, Connecticut 06511, USA}
\author{K. Chou}
\thanks{These two authors contributed equally.  Corresponding author:  kevin.chou@yale.edu}
\affiliation{Department of Applied Physics and Physics, Yale University, New Haven, Connecticut 06511, USA}
\author{C. Shen}
\affiliation{Department of Applied Physics and Physics, Yale University, New Haven, Connecticut 06511, USA}
\author{M. Reagor}
\affiliation{Department of Applied Physics and Physics, Yale University, New Haven, Connecticut 06511, USA}
\affiliation{Rigetti Quantum Computing, 775 Heinz Ave, Berkeley, CA 94710}
\author{C. Axline}
\affiliation{Department of Applied Physics and Physics, Yale University, New Haven, Connecticut 06511, USA}
\author{R. T. Brierley}
\affiliation{Department of Applied Physics and Physics, Yale University, New Haven, Connecticut 06511, USA}
\author{M. P. Silveri}
\affiliation{Department of Applied Physics and Physics, Yale University, New Haven, Connecticut 06511, USA}
\affiliation{Research Unit of Theoretical Physics, University of Oulu, FI-90014 Oulu, Finland}
\author{C. Wang}
\affiliation{Department of Applied Physics and Physics, Yale University, New Haven, Connecticut 06511, USA}
\author{B. Vlastakis}
\affiliation{Department of Applied Physics and Physics, Yale University, New Haven, Connecticut 06511, USA}
\affiliation{IBM T.J. Watson Research Center, Yorktown Heights, NY 10598, USA}
\author{S. E. Nigg}
\affiliation{Department of Physics, University of Basel, Klingelbergstrasse 82, 4056 Basel, Switzerland}
\author{L. Frunzio}
\affiliation{Department of Applied Physics and Physics, Yale University, New Haven, Connecticut 06511, USA}
\author{M. H. Devoret}
\affiliation{Department of Applied Physics and Physics, Yale University, New Haven, Connecticut 06511, USA}
\author{L. Jiang}
\affiliation{Department of Applied Physics and Physics, Yale University, New Haven, Connecticut 06511, USA}
\author{S. M. Girvin}
\affiliation{Department of Applied Physics and Physics, Yale University, New Haven, Connecticut 06511, USA}
\author{R. J. Schoelkopf}
\affiliation{Department of Applied Physics and Physics, Yale University, New Haven, Connecticut 06511, USA}

\date{\today}
\begin{abstract}
There are two general requirements to harness the computational power of quantum mechanics:  the ability to manipulate the evolution of an isolated system and the ability to faithfully extract information from it.   
Quantum error correction and simulation often make a more exacting demand: the ability to perform non-destructive measurements of specific correlations within that system.
We realize such measurements by employing a protocol adapted from [S.\ Nigg and S.\ M.\ Girvin, Phys.\ Rev.\ Lett.\ {\bf 110}, 243604 (2013)], enabling real-time selection of arbitrary register-wide Pauli operators.
Our implementation consists of a simple circuit quantum electrodynamics (cQED) module of four highly-coherent 3D transmon qubits, collectively coupled to a high-Q superconducting microwave cavity.
As a demonstration, we enact all seven nontrivial subset-parity measurements on our three-qubit register.
For each we fully characterize the realized measurement by analyzing the detector (observable operators) via quantum detector tomography and by analyzing the quantum back-action via conditioned process tomography.
No single quantity completely encapsulates the performance of a measurement, and standard figures of merit have not yet emerged.
Accordingly, we consider several new fidelity measures for both the detector and the complete measurement process.
We measure all of these quantities and report high fidelities, indicating that we are measuring the desired quantities precisely and that the measurements are highly non-demolition.
We further show that both results are improved significantly by an additional error-heralding measurement.
The analyses presented here form a useful basis for the future characterization and validation of quantum measurements, anticipating the demands of emerging quantum technologies.
\end{abstract}
\pacs{here}
\maketitle

\section{INTRODUCTION}
Building on impressive progress in control \cite{martini_threshold, high_fid_ion_gates}, measurement \cite{jpc, high_fid_ion_readout}, and coherence \cite{hanhee, clocks}, experimental quantum information science is addressing increasingly complex challenges, such as quantum error correction (QEC) \cite{shor_factoring, gottesman_thesis, surface_code, hanson_qec, martinis-classicalerrorcorrection, nissim_qec} and quantum simulation \cite{universal_quantum_simulators, engineering_quantum_dynamics}.
These applications frequently call for measurements of multi-qubit properties, which can be qualitatively different from one-qubit measurements.
Crucially, measurements of correlations, rather than complete state information, require a more refined concept of non-demolition.
Strong single-qubit measurements project the system into a one-dimensional and trivial subspace, and non-demolition is guaranteed if repeated measurements agree.
This is only a necessary, but not sufficient, condition for measurements of correlations, which must project the system into a multi-dimensional subspace while maintaining coherence within that subspace--an idea with no one-qubit analog.
These measurements must be accomplished such that we learn only the desired information and no more.
Great care is required to engineer this intricate interaction of a complex and delicate quantum system with the noisy and dissipative outside world.

In principle, these measurements can be constructed from a set of primitives consisting of one-qubit measurements and a universal set of one- and two-qubit gates. 
In practice, building up these circuits is not as simple as stringing these primitives together; decoherence and residual interactions play an increasingly important role. 
Residual interactions are of particular concern as they lead to correlated and coherent errors, can scale badly as additional qubits are added, and are potentially catastrophic for quantum error correction \cite{preskill-longrange, fowler-nonlocal, martinis-metrology}.
These challenges raise two questions: how do we design hardware and software to directly implement multi-qubit measurements while addressing these concerns, and how do we reasonably characterize these measurements given different requirements?

Non-demolition multi-qubit measurements have been implemented in a variety of architectures. 
Impressively, measurements of three- and four-qubit properties have been demonstrated in superconducting and ionic systems \cite{ibmbbn-machinelearning, ibm_zzzz, blatt-colorcode, blatt-opensystem}.
Quantification of these measurements has generally consisted of measurements of simple eigenstates of the intended measurement operator.
Characterization of a multi-qubit measurement process, including the back-action, has been reported for $ZZ$ measurements in two-qubit systems, via conditioned process tomography in superconducting qubits experiments \cite{ibm_parity,leo_parity}. 
Generally, the precise implementation of measurements within larger quantum systems is an important direction for further study.  
Superconducting qubits, our chosen platform, have strong electromagnetic interactions which lead to fast and high-fidelity control and single-qubit measurement.
However, without care in engineering their environment, these same interactions can make these qubits vulnerable to crosstalk and decoherence  \cite{martinis-classicalerrorcorrection, ibm-simultaneousRB, leo_parity, gerhard_dipole}.
Among other adverse effects, this crosstalk is also likely to pollute measurements (and through back-action, the system) with extraneous information.

With these issues in mind, in this work we demonstrate a novel 3D cQED architecture which exhibits direct qubit-qubit couplings significantly smaller than the qubit linewidths.
Instead, interactions among our highly coherent and simple qubits are mediated by a common superconducting cavity. 
We use these interactions to engineer measurements of multi-qubit properties via an ancilla qubit, adapting a proposal by Nigg and Girvin \cite{nigg_stabilizer}.
We demonstrate this protocol with a 3+1 qubit system by performing all seven non-trivial three-qubit subset-parity measurements, $O_i \otimes O_j \otimes O_k : O_x  \in \left\lbrace I,Z \right\rbrace$, excluding $III$. 
This set is of particular interest: when combined with single-qubit rotations, it generates the measurements of all possible product operators, which include those needed for stabilizer-based quantum error correction \cite{gottesman_thesis}.

No single number fully characterizes a measurement, and the need for more sophisticated assessment is amplified for larger systems as they admit richer phenomena.
Accordingly, we characterize each demonstrated measurement with three experiments and discuss several figures of merit.
First, we perform an analysis similar to the measurement of computational states and report the assignment fidelity.
Second, we introduce and employ a novel form of quantum detector tomography \cite{detector_tomography} to fully extract the positive operator valued measure (POVM) \cite{mike_and_ike} that describes the realized detector.
We do several analyses based on these results, introducing two fidelity measures and a complementary measure we call the specificity.
Third, we consider the back-action induced by the measurement and reconstruct the measurement process maps using conditioned process tomography.
We further introduce two analogous fidelity measures for the measurement process.
Additionally, many error mechanisms, such as relaxation, leave a distinctive signature, and we also report results heralded by an additional measurement confirming success.
This may prove to be a useful feature, as heralded-success gates can be efficiently used for universal quantum computation \cite{knill_heralded_gates, aliferis_heralded_gates,knill_post_selected}.

\section{Implementing the measurement apparatus}
Our system is centered around a high-Q superconducting cavity (hereafter the ``cavity" and with resonance frequency $f_c$), which is used mechanically as the isolating package for our module and quantum-mechanically as an ancillary pointer state.
Four 3D transmons (with $\ket{g} \leftrightarrow \ket{e}$ transition frequencies $\left\lbrace f_i\right\rbrace$) couple to this cavity, with qubit-cavity dispersive interaction rates $\lbrace \chi_i \rbrace$ and qubit-qubit longitudinal interaction  rates $\lbrace \chi_{ij} \rbrace$.  
The simplified undriven Hamiltonian is given by

\begin{figure}[t]
\centering
\includegraphics[width=3.4in]{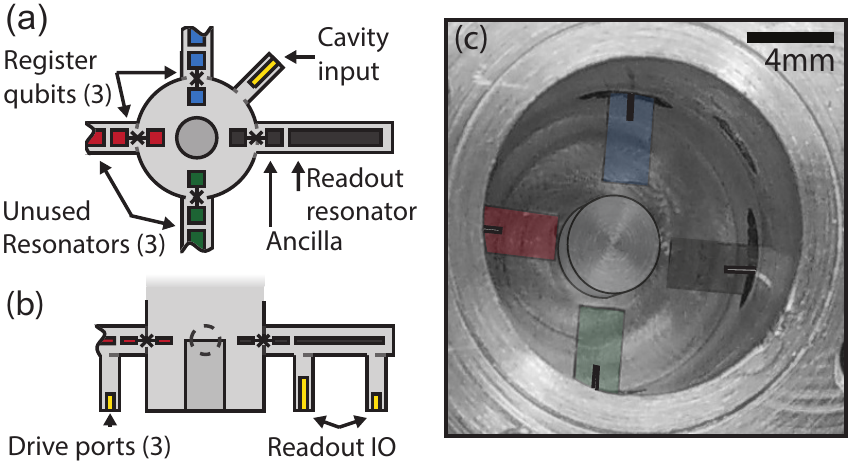}
\caption{\label{fig:2d_cartoon}{
The experimental sample consists of a central $\lambda / 4$ stub resonator \cite{matt_stub}, machined out of 6061 aluminum, with a lifetime of 72 $\mu$s, consistent with our expectation of the limitation due to surface losses.
Four sapphire chips enter the cavity radially, each of which supports a 3D transmon and a quasi-planar coaxial $\lambda / 2$ resonator \cite{coax}, patterned in the same lithographic step. All four $\lambda / 2$ resonators have undercoupled input ports for fast individual qubit control. One resonator has a low-Q ($1/\kappa = 60 \text{ ns}$) output port which leads to a Josephson Parametric Converter \cite{jpc}. This enables high-fidelity ($98\%$) readout of the directly coupled qubit, which we designate as the ancilla. The other three qubits are designated as the register, and their associated three resonators are unused. All qubits share essentially identical capacitive geometry, but differing Josephson energies space the qubits by roughly 400 MHz. This results in dispersive shifts $\lbrace \chi_i \rbrace = \lbrace 1.651, 1.194, 0.811, 0.613 \rbrace$ MHz and $\lbrace\chi_{ij}\rbrace$ generally on the order of 1 kHz.  Further Hamiltonian and coherence details are in the supplement.  The cavity has an undercoupled input port, used for conditional and unconditional displacements, and a diagnostic output (which is also undercoupled and is not depicted).  (a) and (b) depict top- and side-view schematics, respectively.  Not to scale.  (c) False color top-view of the physical device with  outlines for clarity.
}}
\end{figure}

\begin{figure*}[t]
\centering
\includegraphics[width=7.0in]{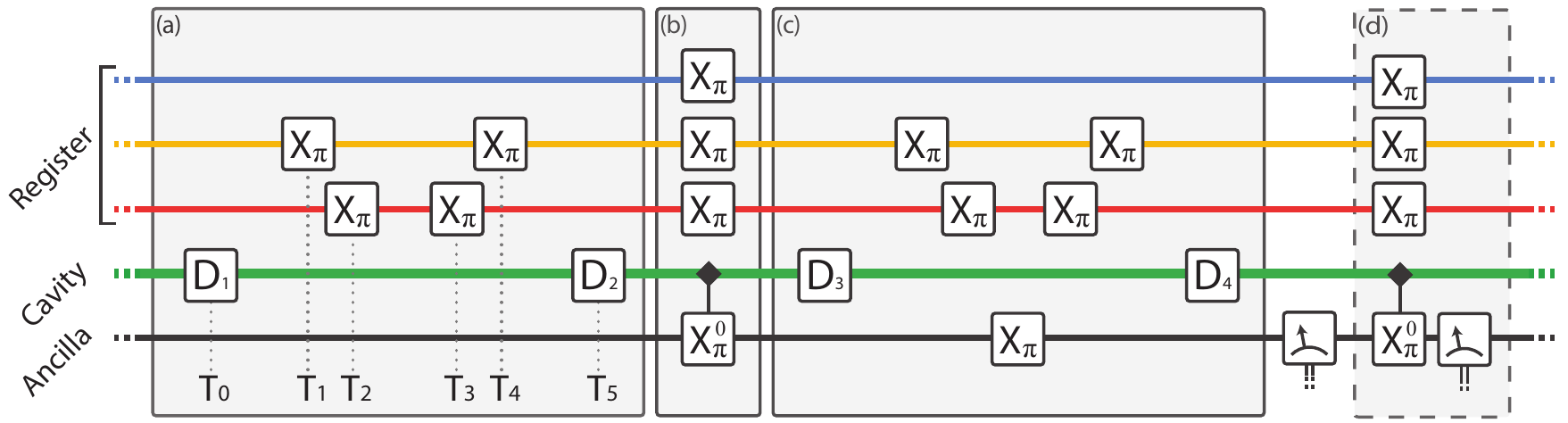}
\caption{\label{fig:ziz_msmt}{
Circuit diagram for ZIZ measurement. Steps $X_{\pi}$ refer to one-qubit rotations around the $X$ axis by $\pi$ radians.  
Steps $X^0_\pi$ indicate that the pulses are spectrally narrow and are roughly selective on having zero photons in the cavity.  Steps $D_{i}$ represent unconditional displacements of the cavity.  The meters are measurements of the ancilla via the readout resonator, which is not itself depicted. 
The ancilla has the largest dispersive shift and the register qubits are then numerically ordered (from top to bottom) such that  $\chi_1 < \chi_2 < \chi_3$.
Prior to this procedure a series of measurements is applied to post-selectively prepare the ground state, see the supplement for details.  
(a) The algorithm begins with a displacement $D_1$ to create a coherent state of $\bar{n} = 5$ photons into the cavity, which acquires a phase shift $\theta$ in a time  $T =T_5 -T_0  \approx \theta / (2\pi \chi_1)$ conditionally on the state of qubit 1 (blue).
For measurements of one- and two- qubit properties $\theta = 2\pi / 5$.
In this example, we perform a full echo on the second qubit (yellow) by performing two unconditional $X$ gates separated in time by $T_4 - T_1 \approx T/2$.  
The third qubit (red) would contribute a conditional phase shift of $2\pi\chi_3 T > \theta$.
We reduce this to $\theta$ by performing two $X_\pi$ gates separated by $T_3 - T_2 \approx \theta \left(\chi_1^{-1} - \chi_3^{-1}\right)/2$.
At $T_5$, we perform $D_2$ to shift the odd two-parity coherent-state pointer to the zero-photon state. 
Note that the overlaps between the even two-parity pointer states and the zero-photon state are exponentially suppressed.
(b) We map this photon number information onto the ancilla qubit with a $X^0_\pi$ gate, taking advantage of the well known number-splitting phenomenon \cite{number_splitting}.  
As the cavity states are separated by $\approx 6.5$ photons, we employ a faster, approximately selective gate, 300 ns in duration.  
$X_\pi$ gates on the register are centered on this pulse in time to echo away the cavity evolution during this step.
(c) To disentangle the cavity pointer states we essentially invert the pulse sequence of (a), returning the cavity to the vacuum state. 
We must also echo the ancilla as it may now be excited.
This results in a total gate length of $970$ ns. Subsequently, we measure the ancilla qubit.
(d) This optional step determines if there are residual photons in the cavity.
Since many types of errors result in residual photons, a subsequent photon-number-selective rotation and measurement of the ancilla heralds these errors.
When measuring three qubits (e.g. $ZZZ$), we choose $\theta = \pi$ so that the cavity states entangled with the one- and three-excitation manifolds recohere.
}}
\end{figure*}

\begin{align}
\label{eq:1}
H/h =& \sum_i f_i \outerproduct{e}{e}_i + f_c a^{\dag}a - \sum_i \chi_i \outerproduct{e}{e}_i a^{\dag}a \nonumber \\
&- \sum_{i, j\neq i} \chi_{ij} \outerproduct{ee}{ee}_{ij},
\end{align}

where we truncate the bosonic modes of the transmons to the two lowest energy levels. 
A more comprehensive Hamiltonian is given in the supplement.  
Our device, depicted and described in Fig. \ref{fig:2d_cartoon}, provides simultaneously strong qubit-cavity interactions and weak qubit-qubit interactions with $\chi_i /  \chi_{ij} \approx 10^3$. 
This architecture also admits independent drive and readout channels, leading to minimal classical cross-talk in both control and measurement. 
In this work we utilize the readout channel (hereafter the ``readout resonator") of only one qubit which serves as the ancilla. 
The other three qubits we collectively refer to as the register.
Our protocol uses the qubit-cavity interactions to map a property of the register state (e.g. $ZZZ$ or $ZZI$) onto the ancilla via manipulation of the cavity state.

In the following, we provide a general overview of our three-step protocol. For more details see Fig. \ref{fig:ziz_msmt}.
First, the cavity mode is displaced from the vacuum. 
It then acquires phase conditionally on the qubit states due to the dispersive interaction.
This evolution is qualitatively akin to a continuous and parallel cPHASE interaction.
As the cavity accumulates phase, we build the measurement operator qubit-by-qubit by using pairs of $X$ gates, analogous to a Hahn echo.
If we want an $I$ in the measurement operator for a given qubit, i.e. not measuring it, we can perform a ``full echo" decoupling sequence and completely average out the phase contribution of a particular qubit unconditionally (qubit 2 in Fig. \ref{fig:ziz_msmt}).
On the other hand, if we want a $Z$ in the measurement operator for a given qubit, then we want it to contribute a specific conditional phase angle to the cavity state. 
We may either allow the natural dispersive evolution to achieve this target conditional phase (qubit 1 in Fig. \ref{fig:ziz_msmt}) or precisely tune this contribution by applying a ``partial echo'' sequence (qubit 3 in Fig. \ref{fig:ziz_msmt}).

In that manner, measuring a qubit, or not, can be chosen by the timing of echoing gates.  
When all of the register qubits being measured give equal phase contributions, the phase of the cavity is a natural meter for the number of excitations in the measured subspace.  
This equalization can be considered as a stroboscopic erasure of which-path information, resulting in entanglement between a selected multi-qubit property of the register and the phase of the cavity. 
The meter can go beyond excitation counting to map other operators; for example, if the accumulated phase per qubit is $\pi$, the cavity state measures register parity regardless of the size of the register.

In the second step (Fig. \ref{fig:ziz_msmt}b), the ancilla qubit samples whether or not the cavity has acquired some chosen phase.   
An unselective displacement shifts one of the cavity states to the vacuum, converting phase information into photon number-state information.  
A spectrally narrow pulse then excites the ancilla if and only if there are zero photons in the cavity.
These first two steps are a natural multi-qubit extension of the ideas used in the qcMAP gate \cite{qcmap}.  
With the chosen property of the register now imprinted onto the ancilla state, we could measure the ancilla directly, but if we hope for the measurement to be non-destructive we must first disentangle the cavity.
In the last step (Fig. \ref{fig:ziz_msmt}c), we remove this residual entanglement, unconditionally resetting the cavity to the vacuum, by essentially inverting (or ``echoing") the unitary dynamics of the first step.  
Finally, we use this composite gate to enact a multi-qubit measurement by interrogating the readout cavity, which is sensitive only to the ancilla state.

We optionally append an additional manipulation and measurement to verify that the cavity has been reset to the vacuum (Fig. \ref{fig:ziz_msmt}d).
This condition is not achieved when we have experienced errors due to qubit or cavity relaxation, as well as certain effects from higher order terms in the Hamiltonian.

\section{MEASUREMENT CHARACTERIZATION}

With these measurements manufactured, we turn to the problem of describing them quantitatively. 
This endeavor does not have a one-size-fits-all resolution: different applications have differing needs, requiring experiments and analyses of differing complexity. 
Accordingly, we attempt to anticipate many potential desiderata, and perform several analyses on the results of three separate experiments.  
The first two experiments examine the detector alone, neglecting back-action on the input state, and the third goes on to examine the measurement process entirely.  

The first and simplest analysis provides a partial characterization of the detector.
\begin{enumerate}[label=\Alph*.]
\item Assignment fidelity:  How much desired information are we getting?  
\end{enumerate}

We expand on this with a full characterization of the detector and extract several figures of merit.
\begin{enumerate}[resume*]
\item Quantum detector tomography: What is the POVM that describes our detector?  What are we actually learning?
\item Specificity: What is the maximal measurement contrast along \emph{any} axis? And a complementary question: how close is this maximal axis to the desired axis?  
\item Detector fidelity: How close is our POVM to the desired measurement?
\end{enumerate}

We additionally may care a great deal about the back-action of the measurement, a consideration crucial for stabilizer-based QEC.
For this, we perform a third experiment to characterize the measurement process and extract pertinent figures of merit based on the quantum instrument formalism \cite{wilde_qi, ozawa_qi, davies_qi}.

\begin{enumerate}[resume*]
\item Measurement process tomography: What makes the detector click \emph{and} what happens to the state after measurement?
\item Quantum instrument fidelity: How close are the measurement processes, inclusive of back-action, to the desired measurement?
\end{enumerate}

As the quantum instrument encompasses the detector, a discrepancy between the quantum instrument and detector fidelities provides an assessment of the undesired back-action on the system.

In order to extract the detector and quantum instrument fidelities, we describe our measurements as channels that introduce the detector as a classical state in an additional Hilbert space. 
We provide more details on this interpretation as these quantities are introduced. 
This treatment allows us to appropriate commonly-used figures of merit for quantum processes and apply them to measurements.
Following the reasoning in Gilchrist et al. \cite{nielsen_process}, for each analysis we report the two sets of measures: first, the \emph{J-fidelity}, which is derived directly from the channel Jamio\l{}kowski matrix, and second, the \emph{S-fidelity}, which is a conservative measure based on the worst-case input state.
The J-fidelities are similar to measures given in \cite{dressel_measfid, magesan_measurements}.
We provide the ranges for all of these results in the main text and tables of the full results in the supplement. 
In addition to the fidelities, we also provide the analogous J- and S-distances, the latter of which is commonly called the diamond distance \cite{kitaev_qec,mike_and_ike}.

\subsection{Assignment fidelity}
A simple way to define the performance of a binary measurement is to assume the model, i.e. how sensitive is our measurement to the desired quantity (e.g. $ZIZ$)?
To answer this, we prepare states of known ideal measurement outcome, measure, and fit to find the correlation between the state preparation and experimental outcomes.
This yields a contrast and an offset, which can also be interpreted to tell how often we get the expected result.
As an example, for a simple $Z$ measurement of a qubit, it is common to prepare the computational ($\ket{0}$ and $\ket{1}$) states and to report how often the measurement outcomes agree with the state preparation.  
This is often referred to as assignment fidelity. 
The extension to multi-qubit subset-parity measurements typically involves preparations of $d$ computational states, e.g. \cite{ibm_zzzz}.

We show the results for our measurements with a similar but slightly more illustrative experiment in Fig. \ref{fig:trirabi}.
We prepare a larger number of states than required, but the result is essentially the same as we fit the data to our expected correlation.
We find contrasts well over $90\%$ for most of our measurements, and in some cases approach the limit set by our ancilla readout, showing that our detectors are highly sensitive to the operator we expected.
Assignment fidelity is a useful diagnostic tool as it is quickly measured and has a simple interpretation, however it provides limited information.
As an exaggerated example, if you expected a $Z$-sensitive detector, this present analysis would give identical results if used to examine either a random number generator or a perfect $X$-sensitive detector.
Errors like the latter, if undiscovered, would result in misleading, skewed state estimation, but are correctable with unitary control. 
A simple noisy reduction of contrast has neither of these properties.

\begin{figure}[h]
\centering
\includegraphics[width=3.4in]{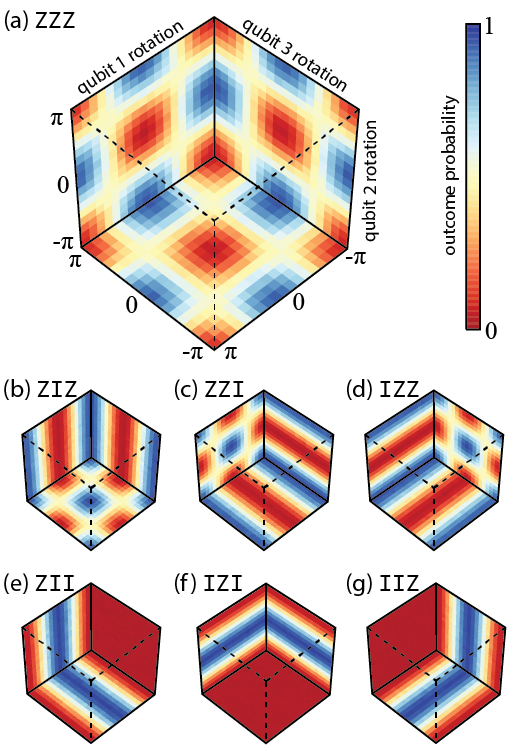}
\caption{\label{fig:trirabi}{
Demonstration of parity measurement outcomes.
For all seven non-trivial three-qubit subset parity operators we show ancilla excitation probabilities as a function of initial register state.  
In each subfigure the three axes specify the initial state of each register qubit, parameterized by rotation angle about the $X$ axis after initialization in the ground state.  
We depict three plane-cuts through that parameter space.  
The color scale indicates the probability to find the ancilla in the excited state.
The top figure (a) shows measurement operator ZZZ. 
The second row (b)-(d) shows two-qubit parity measurements.  It is easily seen that the outcome is independent of the preparation of one qubit.  
The third row (e)-(g) shows single-qubit measurements, reflecting sensitivity to only one preparation axis. 
We extract assignment fidelities from these data of $89-95\%$ that improve to $94-97\%$ with post-selection on a success herald.
}}
\end{figure}

\subsection{Quantum detector tomography}
More generally we may ask of a binary-outcome detector: for what inputs does it ``click?''
In quantum mechanics we describe these detector-outcome probability distributions with the POVM formalism. 
The measurement is represented as a set of operators $\lbrace E_i \rbrace$ with the probability of measurement outcome $i$ equal to $\Tr\left[E_i \rho \right]$ given an input state $\rho$.
As POVMs are complete, we may fully describe a binary POVM with one operator, $\lbrace E_0, E_1 \rbrace \equiv \lbrace E, I - E \rbrace$.

To characterize the detector more rigorously, we prepare a complete or over-complete set of known input states and record the measurement outcome distribution, or ``click" probability.
We reconstruct $E$ from these data with a linear inversion.
This procedure is called quantum detector tomography \cite{detector_tomography, offdiagonal_detector_tomography}.
It is essentially identical to traditional state tomography, differing only in the prior assumptions: rather than assuming we know the measurement operator, as we do in state tomography, we assume knowledge of the input state.
This knowledge is imperfect, but we note that our ground state preparation is better than $99\%$, and our one-qubit gate errors are less than $0.2\%$ as determined from randomized benchmarking \cite{knill_rb}. 
Our implementation of detector tomography is the first that we know of outside of photonic experiments for system dimensions greater than two, and the first we know of at all in superconducting systems. 
We also note that a weaker, diagonal form of detector tomography that assumes sensitivity only to $I$ and $Z$ correlations, akin to our assignment fidelity analysis, is often implicitly used to correct for measurement errors in state tomography \cite{qbus}.

\begin{figure*}[t]
\includegraphics[width=7in]{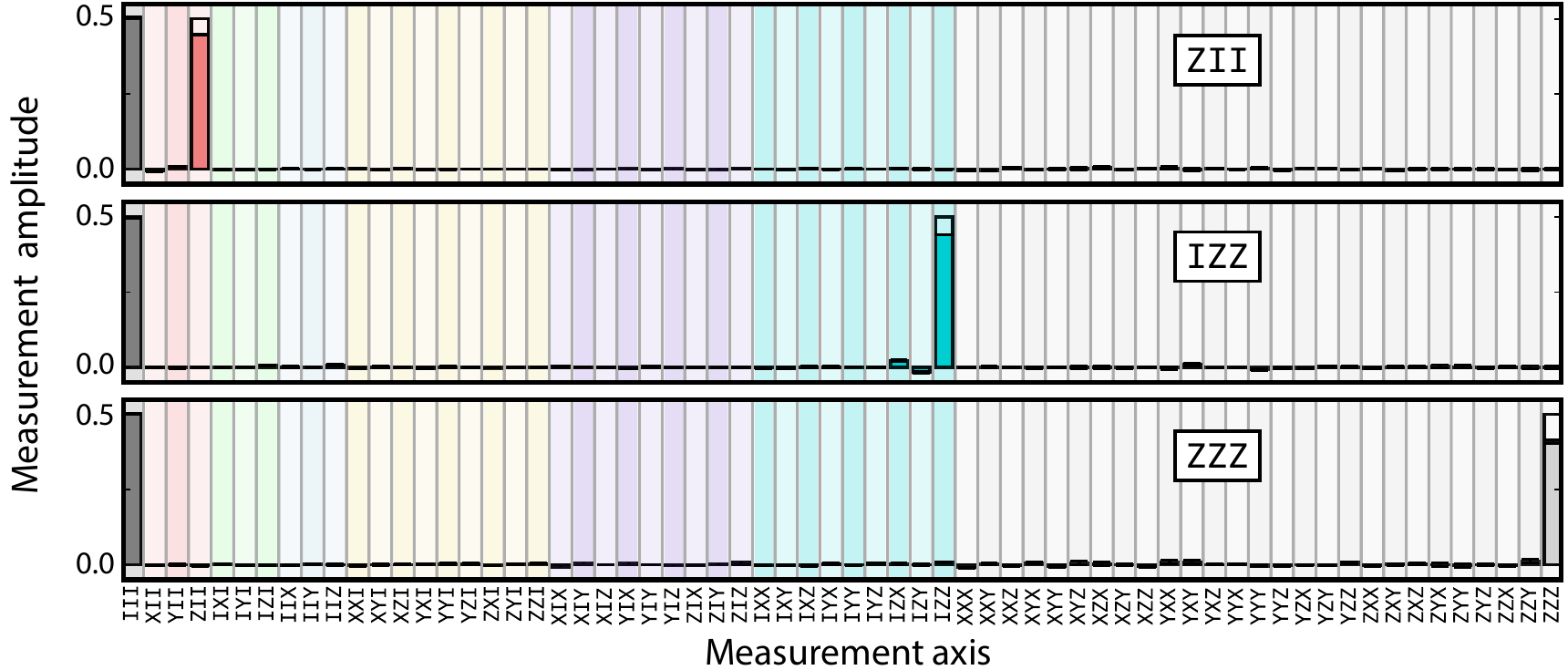}
\caption{Results of quantum detector tomography for three selected operators, using the unheralded datasets.  We expand the first element $E$ of each POVM in three-qubit generalized Pauli operators $\sigma_i$, so that $E=\sum_i c_i \sigma_i$, and show the magnitudes of the coefficients of that expansion. For measurement of a Pauli operator, each should have two non-zero bars (amplitude 0.5) corresponding to the identity and the operator of the measurement, $\sigma_m$. Deviations of the identity bar from 0.5 indicate that the meter has some bias in the detector outcome distribution.  When the amplitude of the $\sigma_m$ bar is less than 0.5, it indicates the measurement does not have full contrast along the desired axis. Finite values of the other bars indicate that our measurement has undesired sensitivity to an extraneous property.
The POVM J-fidelities for the illustrated operators are $95\%$, $94\%$, and $91\%$ respectively.
The other four realized measurement operators, as well as reconstructed POVMs from the success-heralded dataset, are provided in the supplement.
\label{fig:povms}
}
\end{figure*}

For the seven demonstrated measurements $\lbrace ZZZ$, $ZZI$, $ZIZ$, $IZZ$, $ZII$, $IZI$, $IIZ \rbrace$  we reconstruct the relevant measurement operators, which are ideally projectors onto a subspace of definite measurement outcome.  
Several examples are shown in Fig. \ref{fig:povms}, and we see that the results are close to the expected operators.
We can also distill this full measurement operator into more easily interpretable figures of merit.

\subsection{Specificity}
What is the maximum information our measurement gains about \emph{any} quantity?
Or from another point of view: is our detector infidelity due to noise and simple lack of contrast, or is it because our detector is measuring the wrong quantity?
Furthermore, is our detector biased--given a completely mixed input state, is one measurement outcome more likely than another?
This analysis is general to all strong binary measurements, but here we assume for simplicity that the ideal measurement is of a target Pauli operator, $\sigma_T$.

We express $E$ in the Pauli basis (which fully spans the space of $n$-qubit observables) as in Fig. \ref{fig:povms}, leading to a vector of Pauli coefficients.
One dimension will correspond to the identity axis, one to the desired operator, and the rest to the various other Pauli operators.
We can then find a new basis which rotates this vector space, leaving the identity and $\sigma_T$ axes invariant, such that only one other non-zero coefficient in the measurement vector remains.
This additional coefficient $c_\text{O}$ corresponds to an orthogonal rotated Pauli $\sigma_O$,

\begin{equation}
\label{eq:POVM_specificity}
E = c_I I + c_T \sigma_T + c_O \sigma_O.
\end{equation}

The coefficients of this expansion are easily interpreted: a deviation of $c_I$ from the ideal value of $0.5$ describes the bias of the detector, $2c_T$ represents the maximum possible gain of information about the quantity we wish to measure, and $2c_O$ represents the magnitude of the potential \emph{undesired} information gain.  
Note that $E$ must be a positive matrix, which yields constraints on $c_T$ and $c_O$ relative to the bias term $c_I$, e.g. a detector that always clicks cannot yield useful information. 
The infidelity corresponding to measuring along the wrong axis can be considered as analogous to coherent errors for standard processes, since a unitary rotation of the system prior to measurement would remove it.
This formulation leads naturally to a description in terms of angles between vectors, and we can quantify how well we are measuring along the correct axis with an angle for the ``specificity'' of the measurement, $\theta_s \equiv \arctan \left(c_O / c_T\right)$.
We find our measurements to be within $1-5 ^{\circ}$  of the target operator, indicating that our measurements are not yielding significant information about unwanted quantities. 
In addition, we see that our detectors have very little bias.

Returning to the initial question of maximal information gain, we recast Eq. \ref{eq:POVM_specificity} as $E = c_I I + c_{\text{max}} \sigma_{\text{max}}$ where the operator $\sigma_{\text{max}}$ is defined as the axis where the measurement gains the most information. 
The coefficient $c_{\text{max}} = (c_T^2 + c_O^2)^{1/2}$ directly quantifies the total information gain along this axis. 
As might be expected from our small $\theta_s$, we find $c_{\text{max}} \approx c_T$ for our measurements.

\subsection{Detector fidelity}

Generally, how similar are two detectors?
We extend this standard POVM formalism by describing the detector as a quantum channel \cite{wilde_qi}.
This detector channel takes a quantum state of the register $\rho_r$ as input and yields a diagonal density operator $\rho_d$ with entries that represent the detector outcome probabilities,

\begin{equation} \label{eq:detector_process}
\mathcal{E}_{\text{det}}: \rho_r \mapsto \rho_d = \sum_i \Tr \left[E_i \rho_r\right] \outerproduct{i}{i}_d.
\end{equation}

The output state of the detector $\rho_d$ is represented as a density operator but should be understood as a container for a purely classical probability distribution.  
It represents the recorded measurement outcome, \emph{not} the state of our physical ancilla transmon.
The detector channel $\mathcal{E}_\text{det}$ is a completely-positive trace-preserving map which can be described with non-square Kraus operators. 
The channels relevant to the present experiment act on an eight-dimensional (three-qubit) register and yield a two-dimensional (binary-outcome) classical detector state.

We compare the experimental detector channels to the ideal processes with two figures of merit. We start with the J-fidelity,

\begin{equation}
\mathcal{F}_{\text{J}} \left( \mathcal{E}^{(1)}, \mathcal{E}^{(2)} \right) 
\equiv 
\mathcal{F}_{\text{Tr}} \left( J^{(1)}, J^{(2)} \right),
\label{eq:Fid_J_det} 
\end{equation}

where $J^{(i)}$ is the Jamio\l{}kowski matrix representing the process $\mathcal{E}^{(i)}$ and $\mathcal{F}_{\text{Tr}}\left(\rho, \sigma \right) \equiv (\Tr \sqrt{ \rho^{1/2} \sigma \rho^{1/2} } )^2$.
For our measurements, we calculate detector J-fidelities between $91-95\%$, which improve to $95-98\%$ with post-selection on the success herald.
This error detection is efficient as the fraction of experiments discarded is similar to the improvement in fidelity. 
These measurements approach the limit set by our ancilla readout, $98\%$.
The detector J-fidelity can be similar to assignment fidelity, but is more general.
It is applicable to measurements with more than two outcomes and allows for comparison of less than full-contrast operators.
Unlike some comparable figures of merit, detector J-fidelity does not require renormalization of the POVM operators, which may discard information.

For standard non-measurement processes, the worst-case performance is often more important than the J-fidelity, and accordingly we also report the S-fidelity \cite{nielsen_process}.
This measure provides a conservative bound on the detector performance, including any possible degradation when the measurement is applied to a subspace within a larger, entangled quantum system.
This S-fidelity applied to the detector process is given by

\begin{align}
\mathcal{F}_{\text{det}} & \left(  \mathcal{E}_\text{det}^{(1)},  \mathcal{E}_\text{det}^{(2)}  \right) 
\equiv \nonumber \\
& \min_{\rho_{ra}} \mathcal{F}_\text{Tr} \Big( \mathcal{E}^{(1)}_{\text{det}} \otimes I \left(\rho_{ra}\right), \mathcal{E}^{(2)}_{\text{det}} \otimes I\left(\rho_{ra}\right)\Big), \label{eq:stab_fid} 
\end{align}

where $\rho_{ra}$ represents a joint pure state $\outerproduct{\psi}{\psi}_{ra}$ of the register and an ancillary, potentially entangled space, e.g. the rest of a quantum computer. 
We emphasize that the states resulting from the process and compared on the right hand side are states of the detector and ancillary space.
In the special case of two-outcome POVMs, we believe the state that reveals the worst-case performance will always be separable, indicating that $\mathcal{F}_{\text{det}}$ is inherently stable.
In the supplemental material we give a proof of this in the case of the two-outcome detector S-\emph{distance}, and we have numerical evidence suggesting that the two-outcome detector S-\emph{fidelity} has the same property.
We stress that the minimization is still useful as it yields the worst-case performance.
Additionally, it can be easily shown that the $\mathcal{F}_{\text{det}}$ reduces to a minimization of the (square of the) \emph{classical} fidelity of the probability distribution of the detector outcomes.
The minimization over input states is performed using a semidefinite programming package in {\sc MATLAB} \cite{cvx_software, cvx_book}.
With this measure, the seven detector fidelities we report are between $88-95\%$, which improve to $93-97\%$ with post-selection on the success herald.  
We see that for our realized detectors the J-fidelities are $1-3\%$ better than the worst-case performance.

\subsection{Measurement process characterization}

We have now characterized the behavior of the detector, but what happens to the input state after a measurement result is recorded?  
What is the back-action of registering a ``click'' or no ``click''?
Quantum error correction, for example, demands that the measurements must be highly quantum non-demolition in the sense of a von Neumann measurement. 
When an ancilla measurement indicates an outcome, e.g. that the register has positive $ZIZ$, the quantum process performed is ideally a projector on to the specified subspace. 
One method of analysis is to describe this process by two trace-non-preserving maps on the register Hilbert space, $\lbrace F^0, F^1 \rbrace$. 
We quantify these maps by performing outcome-dependent quantum process tomography, which has been previously demonstrated for two-qubit measurements \cite{ibm_parity, leo_parity}.  

This reconstruction begins by preparing a complete set of initial register states.
For each initial state we perform our measurement, record the outcomes, and perform state tomography conditioned on those outcomes.
We employ a maximum likelihood estimation (MLE) for each state tomogram, then weight the outcome states by the probability that each measurement outcome was observed.
This subtlety leads to individually trace-non-preserving maps.
To extract the process, we then perform an additional MLE relating the input states to both sets of output states, constraining the full measurement channel to be positive and trace-preserving.  
It would be preferable to perform a single MLE, rather than two, but this problem is computationally imposing in the three-qubit case. 
The sequential approach has the additional benefit of allowing us to recalibrate drifts in our tomographic measurement operator throughout the several hours of data acquisition.

The reconstructed conditioned maps for the $ZZZ$ operator are partially shown in Fig. \ref{fig:qpt}.  
Ideal measurements of generalized Pauli operators have four real elements in the $\chi$ matrix representation. 
Each of these has amplitude $1/4$, with positive diagonal elements and off-diagonal elements which change sign between the two outcomes.
We find good qualitative agreement between this and our experimental data with a small decrease in contrast, indicating that the back-action is close to the ideal von Neumann projections.

\begin{figure}[ht!]
\centering
\includegraphics[width=3.49in]{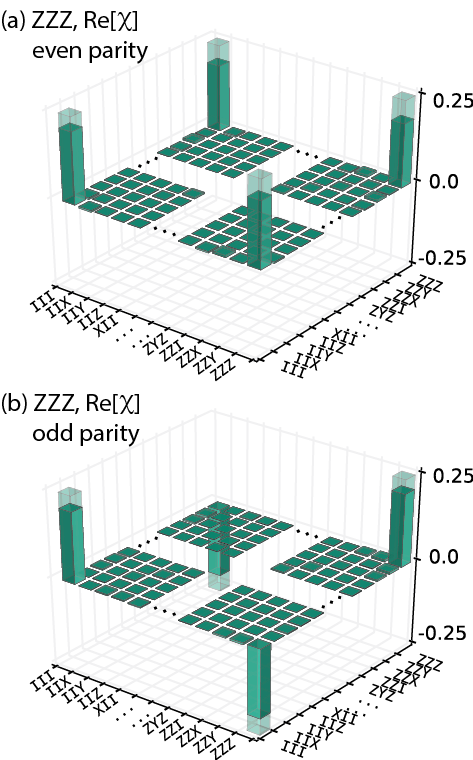}
\caption{\label{fig:qpt}{
Three-qubit conditioned quantum process tomography.
Experimental quantum process tomography (QPT) results for (a) even and (b) odd outcome process maps for the three-parity measurement, $ZZZ$. 
We express our process tomography in the Pauli basis where the conditioned processes can be described using $\chi$ matrix notation: $F^{0(1)} \left(\rho_r\right) = \sum_{ij} \chi_{ij}^{0(1)} \sigma_i \rho_r \sigma_j$ where $\lbrace \sigma \rbrace$ are the three-qubit generalized Pauli operators.  Here we show only the corners of these process matrices, all other parts are visually indistinguishable from noise. Data on the full reconstruction, including of the other six measurement operators, are given in the supplement. The ideal even and odd outcome processes are projectors $\Pi^{0(1)} = \left( III \pm ZZZ \right)/2$, and the corresponding $\chi$ matrices have a simple form consisting of only four real components in the generalized Pauli basis, and this ideal form is overlaid with wireframe bars. 
Note that we plot only the real components as all experimental imaginary components are visually indistinguishable from noise. 
We calculate the J-fidelity (as defined in section \ref{section:F_qi}) for this operator to be $80\%$.
}}
\end{figure}

\subsection{Quantum instrument fidelity} \label{section:F_qi}

How can we compare the performance of our experimentally-reconstructed measurement process to the ideal measurement?
We answer this question by representing the full measurement process as a quantum instrument \cite{wilde_qi}, an approach that parallels our previous treatment of POVMs as channels and shares many of its advantageous properties.
A quantum instrument describes a single channel that takes in a quantum state of the system and outputs both a quantum state of the system as well as a detector-outcome state. 
This detector-outcome state is the same as in Eq. \ref{eq:detector_process} and signals the conditioned back-action induced on the system input state.
For a two-outcome measurement, the quantum instrument can be written as

\begin{align}
\mathcal{E_\text{QI}} : \rho_{\text{r}} &\longmapsto F^0\left(\rho_{\text{r}}\right) \otimes\rho^0_\text{d} + F^1\left(\rho_\text{r}\right) \otimes\rho^1_\text{d}, \label{eq:process_eq_actual} \\
& \stackrel{\text{ideal}}{\longmapsto} 
\Pi^0\rho_\text{r}\Pi^0 \otimes\rho^0_\text{d} + 
\Pi^1 \rho_\text{r} \Pi^1 \otimes \rho^1_\text{d}, \label{eq:process_eq_ideal}
\end{align}

where $\rho_r$ refers to the input state in the register space, $\rho^{0(1)}_d$ refer to diagonal states in the detector subspace $(\outerproduct{0}{0}_\text{d}, \outerproduct{1}{1}_\text{d})$, and $\Pi^{0(1)}$ are projectors in the register space onto orthogonal measurement outcomes.
Similar to our previous analysis of the detectors (without back-action), we derive fidelity measures for these quantum instrument channels.
The J-fidelities for quantum instrument channels follow from Eq. \ref{eq:Fid_J_det} with the appropriate quantum instrument maps and are calculated between $67-75\%$.  These improve to $79-83\%$ with post-selection on the success herald. 
We also report the S-fidelity for these channels, 

\begin{align}
\mathcal{F}_{\text{QI}}& \left( \mathcal{E}^{(1)}_\text{QI} ,  \mathcal{E}^{(2)}_\text{QI} \right) 
\equiv \nonumber \\ 
&\min_{\rho_{ra}} \mathcal{F}_{\Tr} \Big( \mathcal{E}^{(1)} \otimes I \left(\rho_{ra}\right), \mathcal{E}^{(2)} \otimes I \left(\rho_{ra}\right)\Big), \label{eq:qi_fid}
\end{align}

following the same notational caveats as Eq. \ref{eq:stab_fid}.

For our experimental results, we calculate $\mathcal{F}_{\text{QI}}$ of $57-64\%$, which increase to $69-76\%$ with post-selection on the success herald.
We see that the worst-case performance is as much as $10\%$ worse than the J-fidelity measure.  
Note that for the success-heralded data, we include the requisite selective rotation and measurement in the definition of the process, which exposes the register to another 1.5 $\mu$s of decoherence, though we do employ a Hahn echo on all qubits in the register. 
This check could be made significantly faster with an additional, dedicated qubit with a large dispersive shift.

\section{DISCUSSION}

We have shown that the implemented multi-qubit measurements behave as intended: they are highly specific to the desired operator and have detector fidelities that approach the bound set by our single-qubit measurements. 
Additionally, we see that the back-action of the measurement is indeed close to the ideal, but unsurprisingly it is worse than our measurement contrast.

The performance demonstrated here is limited by several effects.
The largest source of infidelity for both the detectors and processes is cavity photon loss, which can be greatly reduced by moving to high-purity etched aluminum.
This is known to give a factor of 10-20 improvement in quality factor in similar systems \cite{matt_stub}, and in similar samples has also increased qubit relaxation times \cite{reinier}. 
The second-largest imperfection in the process performance is the dephasing of one of the register qubits, which displayed a significant low-frequency beat in Ramsey experiments during this experimental run. 
Both detector and process suffer from the low $T_1$ of only 20 $\mu$s for the ancilla, which is limited by the Purcell effect and which may be improved with a Purcell filter \cite{coax}. 
The next-largest imperfections for the detector performance are the relaxation rates of the register qubits.
Additional significant sources of error are the finite cavity anharmonicity and the difficulty in doing photon-number unconditional $X$ gates when photons are present. 
These two effects, as well as other coherent errors, may be circumvented by engineering a sequence via optimal control techniques \cite{grape}, which is feasible for small modules that have weak interactions with their environment. 
The process J-fidelities we obtained are consistent with numerical simulations \cite{qutip} that include all known error sources.

We are unable to directly compare our results to preexisting results in the field, but the most related experimental work \cite{ibm_parity,leo_parity} and theoretical work \cite{magesan_measurements,dressel_measfid} cite fidelity measures that are similar to our quantum instrument J-fidelity (or its square root).
We reiterate that our quantum instrument J-fidelities are as much as $10\%$ higher than our quantum instrument S-fidelities.
In many cases the worst-case performance is of greater importance, as in fault-tolerance considerations.
Indeed, we believe that it may be interesting to incorporate $\mathcal{F}_{\text{QI}}$, or the analogous quantum instrument diamond distance, into threshold calculations, since it directly provides a conservative performance estimate of the operation central to stabilizer-based error correction.
As a more ``compiled'' operation, the full measurement may be more sensitive to non-idealities than concatenated fidelity estimates of smaller one- and two-qubit operations.

A more traditional $ZZZ$ measurement would consist of three \textsc{CNOT} gates and a single-qubit measurement of an ancilla.
Measures derived from the quantum instrument not only take the performance of those simple operations into account, but also account for negative effects due to decoherence and residual interactions among the rest of the qubits.
This includes the duration of the ancilla measurement, which by itself often has unintended and detrimental effects on other qubits in some systems.
Though this analysis requires reconstruction of the process, it may not be overly burdensome for several-qubit stabilizers if compressed sensing techniques are employed \cite{PhysRevLett.105.150401}.

\section{CONCLUSION}

We have demonstrated a versatile quantum gate ideally suited to our highly coherent 3D cQED architecture and used it to enact high-fidelity and specific multi-qubit measurements.  
It is possible to realize strong interactions between fixed-tuned qubits in this system, despite having low direct couplings, and there are clear pathways to further improved performance.
As quantum systems continue to grow in complexity, it will be crucial to build systems with low cross-talk and residual interactions like the one demonstrated here.
We have also presented several new approaches to analyze and characterize measurements, including two new conservative measures to quantify the fidelity of detectors and measurement processes.
Complex measurements within larger systems will become increasingly important, and the figures of merit introduced in this work may prove to be useful tools to benchmark their performance.

\section{ACKNOWLEDGEMENTS}

We acknowledge Reinier Heeres and Philip Reinhold for helpful discussions and software assistance.  We acknowledge Katrina Sliwa, Anirudh Narla, and Michael Hatridge for providing our JPC and advising on its installation and use.  We acknowledge Wolfgang Pfaff and Victor Albert for helpful discussions.
Facilities use was supported by the Yale SEAS cleanroom, YINQE and NSF MRSEC DMR-1119826. This research was supported by the Army Research Office under Grant No.\ W911NF-14-1-0011 and by Office for Naval Research under Grant No.\ FA9550-14-1-0052. C.\ A.\ acknowledges support from the NSF Graduate Research Fellowship under Grant No.\ DGE-1122492. S.\ M.\ G.\ acknowledges additional support from the NSF under Grant No.\ DMR-1301798.  M.\ S.\ acknowledges additional support from the Alfred Kordelin Foundation.  S.\ E.\ N.\ acknowledges support from the Swiss NSF.

\clearpage
\onecolumngrid
\setcounter{figure}{0}
\setcounter{section}{0}
\begin{center}
\textbf{Supplementary Material}
\end{center}

\section{Hamiltonian details}
We give a more comprehensive Hamiltonian in Eq. \ref{eq:fullH}.  We label the ancilla either directly by name or abbreviated as qubit A, and we refer to the register qubits as qubits B, C, or D. The experimentally found values, presented in Fig. \ref{fig:fullH_fig}, are generally within $10\%$ of values predicted from simulation in Ansys HFSS and black box quantization \cite{bbq}. 

\begin{align}
\label{eq:fullH}
H_\mathrm{qubit}/h &= \sum_{i} f_i \outerproduct{e}{e}_i + \left(2f_i - \alpha_i\right) \outerproduct{f}{f}_i \\
H_\mathrm{cavity}/h &= f_c a^{\dag}a - \frac{K}{2} a^{\dag}a^{\dag}aa\\
H_\mathrm{interaction}/h &= -\sum_i \chi_i \outerproduct{e}{e}_i a^{\dag}a - \sum_{i, j\neq i} \chi_{ij} \outerproduct{e}{e}_i\outerproduct{e}{e}_j -\frac{\chi_i'}{2}\outerproduct{e}{e}_i a^{\dag}a^{\dag}aa\\
H &= H_\mathrm{qubit} + H_\mathrm{cavity} + H_\mathrm{interaction}
\end{align}

The qubit frequencies were identified via Ramsey oscillations, and the cavity frequency was found by displacement, waiting, and a second opposing displacement with varying phase. The periodicity of the probability of the cavity having zero photons then is an interferometric indicator of the detuning. 
Repeating this experiment at various displacement amplitudes yields both the Kerr $K$ and the bare frequency $f_c$.  
The frequencies, $\chi_i$ shifts, and $\chi_{ij}$ shifts are defined in the main text.
Register qubit $\chi$ shifts are found by performing the cavity resonance frequency experiment with and without a qubit excited.  Repeating this with initial coherent states with different average photon populations yields the number-state dependent dispersive shift (or qubit-state dependent Kerr) $\chi'$ as the slope and $\chi$ as the y-intercept.  The $\chi$ of the ancilla to the storage cavity is found via qubit state revivals in the presence of cavity photons, and a series of these experiments over initial coherent states with different average photon populations also yields $\chi'$.  These experiments are described more fully in the supplement of \cite{qcmap}.

We measure the direct qubit-qubit coupling $\chi_{ij}$ for all six pairs of qubits. In order to find this, we use a version of a Ramsey experiment where the delay time is fixed and we vary the final $\pi/2$ phase to extract damped, phase-shifted oscillations. 
We perform a pair of these experiments on qubit $i$ with qubit $j$ either initialized in the ground state or the excited state. 
We extract the relative phase shift between these two experiments. We perform a series of these paired experiments for different Ramsey delays and fit the data to a line. The slope is $\chi_{ij}$. We show the results of these experiments in Fig. \ref{fig:zetas}. We note that five of the six couplings are on the order of 1 kHz, and the largest $\chi_{ij}$ of 22 KHz is between the the ancilla and a register qubit, the two qubits that have the smallest detuning (and largest $\chi$) to the storage cavity.
We note that the ancilla-register direct couplings are largely unimportant.  As the ancilla state is known at the end of the experiment, the interaction results in deterministic and known phase shifts.

\begin{figure}[h]
\centering
\label{fig:zetas}
\includegraphics[width=3.5in]{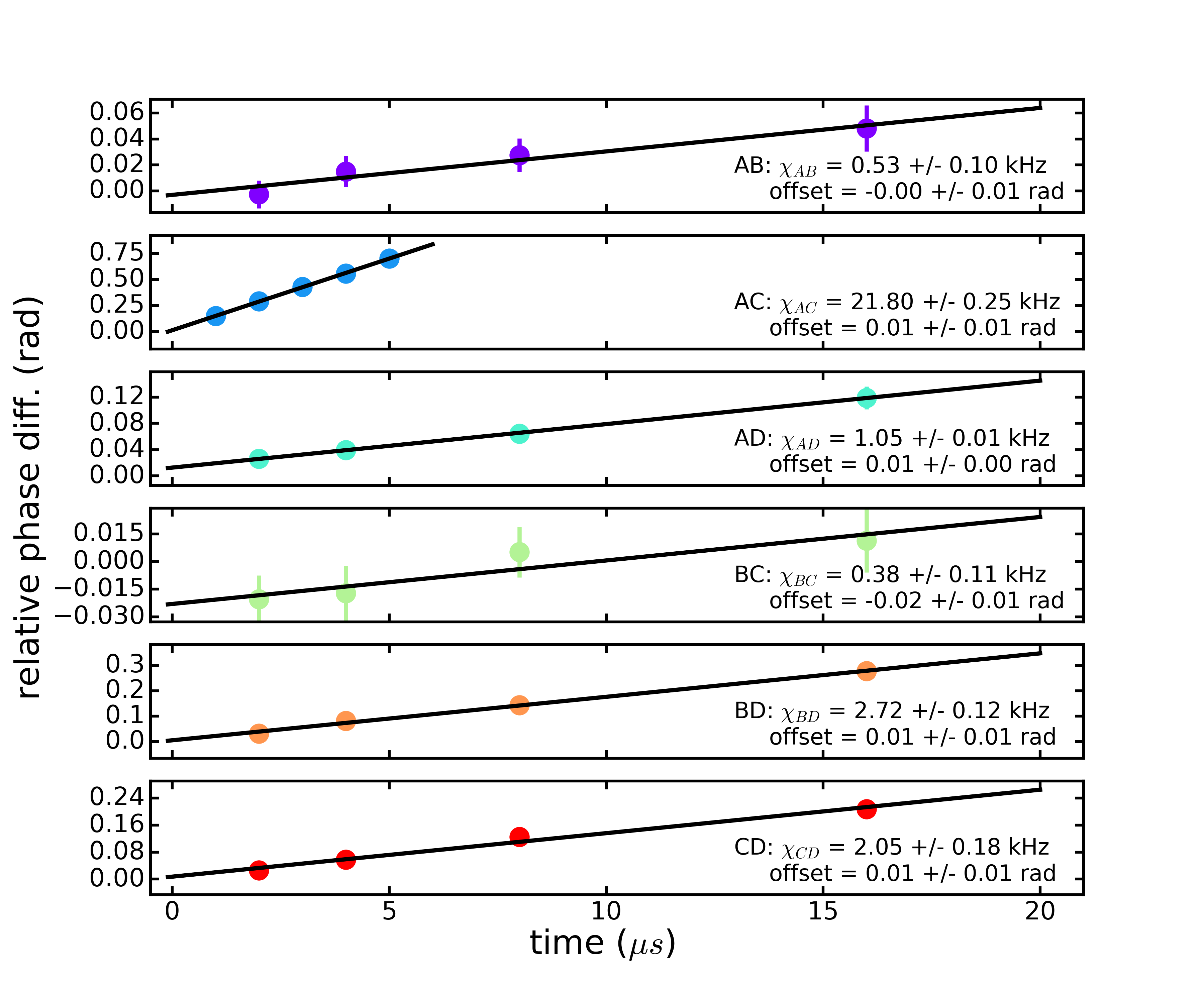}
\caption{Measurement of  $\chi_{ij}$, extracted phase vs. time from Ramsey experiments.}
\end{figure}

The next higher levels of the transmons were probed directly with two-tone spectroscopy to find $\alpha$.  
Note that as we only use direct readout of the ancilla, measurements of the register qubits (e.g. spectroscopy and coherence measurements) are done via a cascaded and non-QND mapping.  
The state of those qubits are probed via a spectrally narrow conditional cavity displacement, followed by using the ancilla to check if the cavity has been displaced from the vacuum.  This yields single-shot measurement fidelities of the register qubits on the order of 0.80.  

Our anharmonicities are on the order of 200MHz, limiting us to timescales for single-qubit operations of roughly 10 ns.  The anharmonicity of the cavity, $K$ for Kerr, is a limiting factor in our algorithm, as it is not reversed under echo. This results in residual photons and dephasing between register states entangled with different photons numbers.  The number-dependent dispersive shift $\chi'_i$, though a sixth-order term (in the expansion of the cosine Josephson energy), is on the same scale as $K$.  It results in similar non-idealities as $K$.  All other higher-order terms do not have a significant effect on this experiment. 

Note that the Hamiltonian of the readout resonator has not been included, however it shares a dispersive shift with the ancillary qubit of roughly 5 MHz, found spectroscopically.  The readout cavity does have a small dispersive `cross Kerr' shift directly with the high-Q cavity of 24 KHz, much smaller than the readout decay rate of 2 MHz.  The cross Kerr was measured via stark shift.  We estimate dispersive shifts with the register qubits on the order of 100 Hz or lower based on additional stark-shift measurements.  These effects may be mediated by the unused planar resonators, which are detuned approximately 50, 350, and 360 MHz from the readout resonator.

\begin{figure}[h]
\centering
\label{fig:fullH_fig}
\includegraphics[width=3.5in]{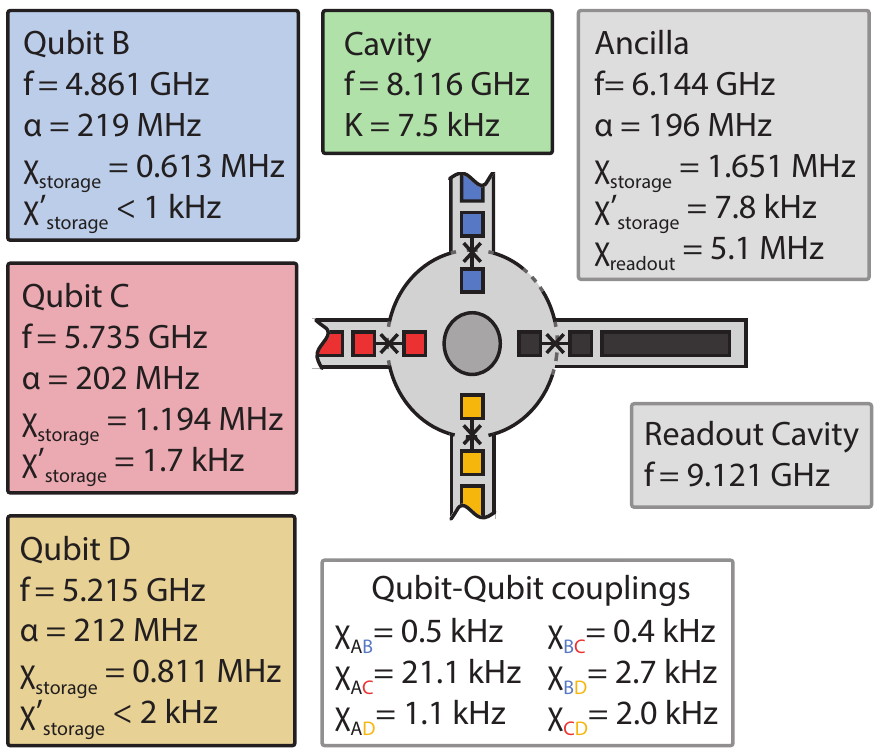}
\caption{Magnitude of Hamiltonian terms in our system.  The functional form of H is given in Eq. \ref{eq:fullH}}
\end{figure}

\section{Coherence Details}
We characterize qubit coherence with three time constants: free energy decay $T_1$, Ramsey oscillation decay $T_{2}^{\text{Ramsey}}$, and Hahn (single) echo decay $T_{2}^{\text{Echo}}$.  In general we use $T_{2}^{\text{Echo}}$ as the relevant coherence time as our algorithm naturally performs several spin echos, and we add echo pulses in steps where they do not occur naturally (e.g. during ancilla measurement and during the optional error detection rotation).  All times are given in microseconds.

\begin{table}[h]
\begin{tabular}{l | c c c }
\hline \hline
qubit & $T_1$ & $T_{2}^{\text{Ramsey}}$& $T_{2}^{\text{Echo}}$ \\ \hline
ancilla & 23(2) & 19(2) & 26(3) \\
register B & 86(5) & 57(8) & 73(10) \\
register C & 87(5) & 62(10) & 77(7) \\
register D & 58(18) & * & 52(8) \\
\hline \hline
\end{tabular}
\end{table}

Note that error in these numbers does not come principally from fitting errors, but rather actual variations in lifetime over time, as is commonly seen \cite{poletto_fluc}.  Qubit D does not have a $T_{2}^{\text{Ramsey}}$ quoted as that Ramsey experiment does not show a single frequency or time constant.  That qubit also has a very large spread in $T_1$ results, which exhibits a non-Gaussian distribution.  In previous thermal cycles qubit D showed clean Ramsey oscillations with lifetimes over 60 $\mu$s.  We are unable to ascribe a cause to this behavior. 

The cavity has an relaxation time of 72 $\mu$s, which is consistent with what we would predict from surface loss for this geometry in 6061 aluminum.   In other experiments with almost identical resonators made from high purity aluminum we see relaxation times well over 1 millisecond.  The readout resonator has a lifetime of 60 nanoseconds ($Q\approx 2000$), which is set precisely by the output coupler.


\section{Single qubit and cavity control}
The single-qubit gates have Gaussian envelopes truncated at $\pm 2\sigma$ and with the derivative of those envelopes played on the opposing quadrature (DRAG).  Unselective rotations on all qubits have $\sigma = 3.5\mu$s.  Unselective cavity displacements are nominally 5 ns square pulses.  The selective rotations on the ancilla come in two lengths.  Truly zero-photon-selective rotations, used in the state preparation, success-herald, and state tomography steps, have $\sigma = 300$ ns.  A roughly zero-photon-selective pulse played in the ancilla-entangling step has $\sigma = 75$ ns.  Both are truncated at $\pm 2\sigma$. We applied single-qubit randomized benchmarking and found single-qubit gate fidelities of $99.8-99.9 \%$.    

\begin{figure}[H]
\centering
\label{fig:rb_data}
\includegraphics[width=7in]{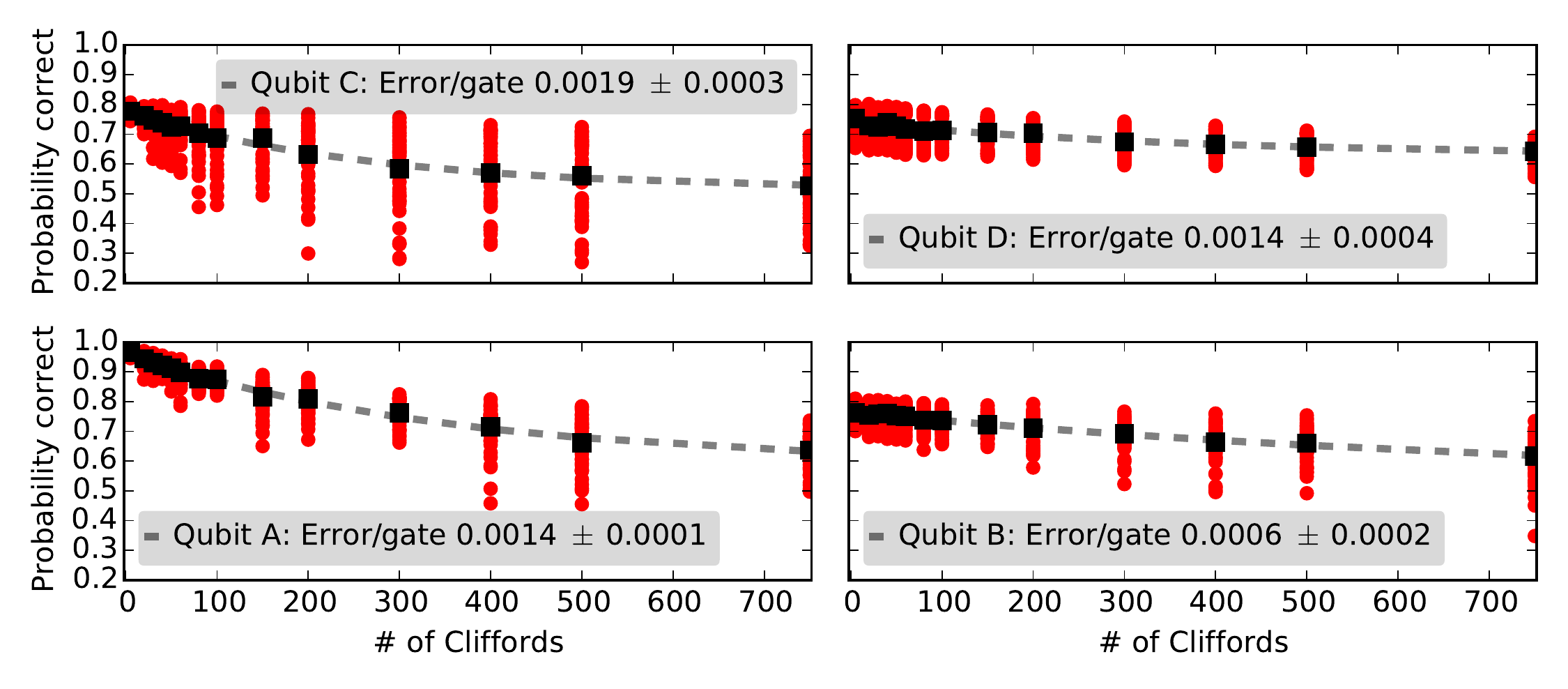}
\caption{Randomized benchmarking results for all four qubits}
\end{figure}

\section{Crosstalk}
\subsection{Direct cross-readout}
Is the measurement through the readout cavity sensitive to the state of the register qubits?  i.e., can we measure them directly?
We estimate this effect by performing a pair of experiments. In both experiments, we induce the same Rabi oscillations on the register qubit under test. In the first experiment, we measure the register qubit with a cavity-mediated readout similar to the state tomography measurement.  In the second experiment, we attempt to measure the register qubit directly via the readout cavity. The first experiment is high-contrast and serves as a calibration experiment. We fit the data to extract the amplitude, frequency, and phase of the Rabi oscillations. As the qubit behavior is nominally identical in both experiments, when fitting the results of the second experiment we constrain the frequency and phase and only allow the amplitude to vary. By comparing the relative amplitudes between the two experiments, we establish a bound on the readout contrast to be in the low $10^{-4}$ level after $100,000$ averages per point. We show the results of this experiment in Fig. \ref{fig:cross_readout}. Note that this experiment is performed with spectrally narrow pulses (Gaussian $\sigma=300$ ns) to rule out cross-driving of the ancilla itself.

\begin{figure}[h]
\label{fig:cross_readout}
\includegraphics[width=7in]{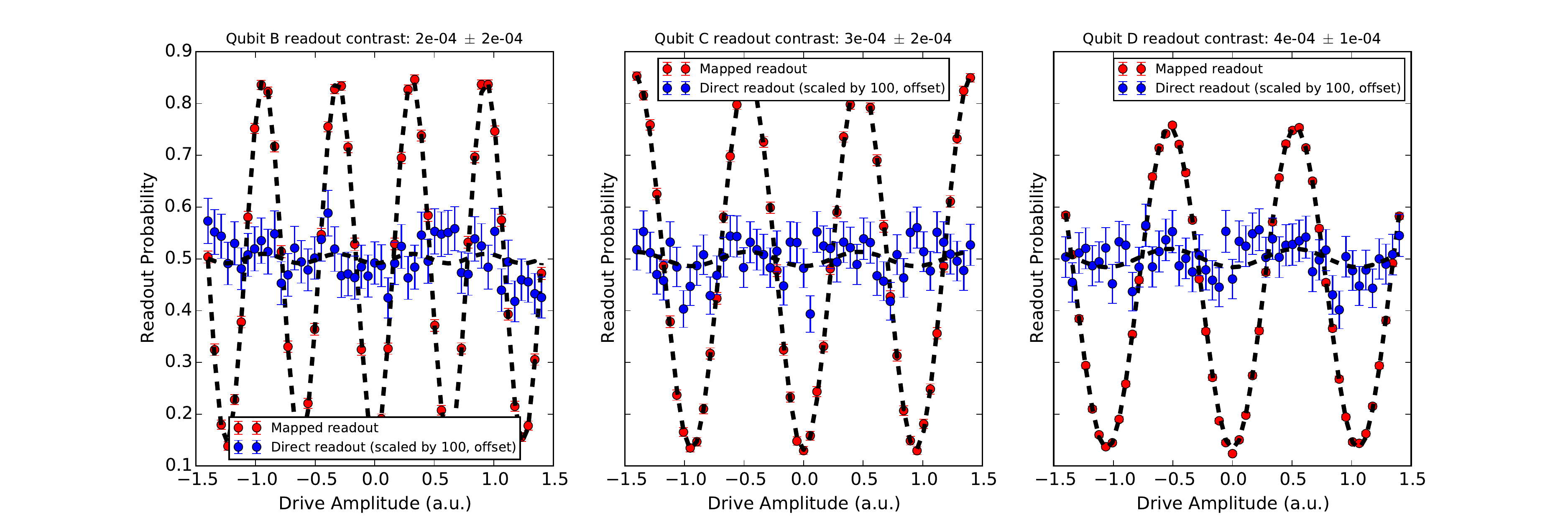}
\caption{Direct cross readout of register qubits}
\end{figure}


\subsection{Simultaneous RB}
We characterize the effect of classical cross-talk in the drives and quantum cross-talk due to any residual interactions using simultaneous randomized benchmarking \cite{PhysRevLett.109.240504}.  The full results are given in table \ref{table:srb_results}.

\begin{table*}[h]
\caption{\label{table:srb_results} Simultaneous randomized benchmarking results.  All quantities are scaled by $10^{-3}$, and the naming conventions follow \cite{PhysRevLett.109.240504}.}
\begin{ruledtabular}
\begin{tabular}{ c c c c c c c c c } 

Qubits (x,y) & $r_x$ & $r_{x|y}$ & $|dr_{x|y}|$ & $r_y $ & $r_{y|x}$ & $|dr_{y|x}|$ & $|d\alpha_{xy}|$ \\
\hline
 $A$,$C$ & 1.4(1) & 1.4(1) & 0.0(1) & 1.9(3) & 1.7(2) & 0.2(4) & 0.2(7) \\  
 $A$,$D$ & 1.4(1) & 1.1(2) & 0.3(2) & 1.4(4) & 0.7(1) & 0.7(4) & 0.3(6) \\
 $A$,$B$ & 1.4(1) & 1.6(1) & 0.2(1) & 0.6(2) & 0.8(1) & 0.2(2) & 0.5(3) \\
$C$,$D$& 1.9(3) & 2.0(1) & 0.1(3) & 1.4(4) & 0.6(1) & 0.8(4) & 0.6(4) \\
$C$,$B$& 1.9(3) & 1.6(1) & 0.3(3) & 0.6(2) & 1.0(1) & 0.4(2) & 0.7(5) \\
$D$,$B$& 1.4(4) & 0.9(1) & 0.5(4) & 0.6(2) & 0.5(1) & 0.1(2) & 0.5(2)
\end{tabular}
\end{ruledtabular}
\end{table*}

\section{Single qubit readout}
The readout resonator has a bare, low-power peak at $f_\textrm{lp} = 9120.78 \textrm{ MHz}$ with a bandwidth of 2.7 MHz.  The ancilla and readout resonator share a dispersive shift of $\chi_r = 5.1$ MHz.  The frequency of our readout tone and JPC are tuned to the average of the two ancilla state-dependent readout frequencies, roughly $f_\textrm{lp} - \chi_r/2$. We readout the state of the ancilla with a square 300 ns pulse that yields a steady-state of roughly 10 photons in the readout resonator. The signal is demodulated and integrated with a matched window, and we threshold the result.  We characterize this performance with three numbers, the probability to get a ground state or `g' result from a second measurement after already getting a `g' result once is $98.9\%$.  The probability of getting an `e' result after getting a `g' and performing a $X_\pi$ pulse is $97.1\%$.  The average, or single-qubit assignment fidelity, is $98.0\%$.  The missing accuracy in the first number is consistent with our outcome histogram overlap, and the additional missing fidelity in the second number is consistent with our $X_\pi$ fidelity and $T_1$ decay during the measurement.

\section{Cavity Q switching details}
Following a similar derivation to \cite{zaki, two_boxes} we use two high power microwave pump tones to drive a four-wave mixing process.  The pumps are detuned from the storage and readout cavities by +200 MHz, yielding a resonant decay time constant of the storage resonator as fast as 500 ns, a factor of $\sim 150$ reduction. Since this process depends on the relative detuning of the drives, we find that this time constant depends strongly on the state of the transmons. We use a chirped, 5 $\mu s$ pump tone to equalize the decay rates. The pump tone envelopes have Gaussian ring-up and ring-down envelope with a sigma of 50 ns.  

This process partially excites the ancilla, so we then perform a similar drive on the ancilla itself, with pumps detuned roughly +200 MHz from the ancilla and readout cavity.  We see coherent swapping between the ancilla and readout cavity with time constant faster than $\kappa$, and choose the first trough in the state-revival oscillations.  This occurs at total pulse time of 250 ns, largely limited by envelope of the pump tones.

\section{Dilution refrigerator setup and experiment electronics}
Microwave pulses on the qubits, cavity, and readout lines originate from CW tones generated by Vaunix LabBrick generators (LMS-802 or LMS-102), depicted by black circles containing vertical sine waves.  The microwaves for the four wave mixing pump tones, readout reference local oscillator, and JPC pump are generated by Agilent PSG and MXG generators, depicted by green circles with horizontal sine waves. The qubit and cavity drives, as well as one of the pump tones, are shaped by single sideband modulation,  using Marki IQ mixers (IQ0618LXP, IQ0618MXP,  IQ0307LXP, and IQ0307MXP).  The envelopes are generated with 65 MHz IF by three Tektronix AWG5014C arbitrary waveform generators.  

With the exception of the readout tone, all inputs are amplified at room temperature by Minicircuits power amplifiers (ZVA-183-S+), then filtered by K\&L low-pass filters (labeled "LPF", part numbers 6L250-10000/T20000-0P/0 or  6L250-12000/T26000-0/0).  The inputs are successively attenuated as they travel down the fridge by cryogenic attenuators.   Parts labeled "Ecco" are home-made, impedance matched low-pass filters, which are coaxial devices filled with Eccosorb CR-110.  The input lines see another K\&L low pass filter, then enter the CryoPerm microwave shield housing the sample, and see another Eccosorb filter.  

While the cavity itself has a (still under-coupled) output line, it is used only for diagnostics and is not relevant to this result.  The output from the readout resonator follows a standard JPC readout chain, being amplified in reflection by the JPC then passed on to a bias-tee and HEMT amplifier at 4K.
Our JPC pump tone is filtered similarly to the other input lines, however it sees a 20 GHz cutoff K\&L filter, as it is at 14.4 GHz.

Our ancilla readout is interferometric to compensate phase drifts in the readout lines.  We use two generators: one at the readout frequency $f_\textrm{readout}$, and the other $f_\textrm{ref}$ that is 50 MHz higher and acts as a reference oscillator. The signal at $f_\textrm{readout}$ is split in two. One arm is used for readout; it is gated by a signal from the AWG, and the tone travels to the readout resonator. The other half is mixed with the signal at $f_\textrm{readout}$ to generate a reference at 50 MHz. The signal that returns from the readout cavity is amplified by a Miteq amplifier (AFS4-08001200--10-10P-4) and is also mixed with this reference oscillator to create a 50 MHz signal. The 50 MHz signal and reference are both amplified (twice each) by a Stanford Research Systems amplifier (SR445A), then digitized by a two-channel 1GS DAC (Alazar ATS9870). The phase of the signal is shifted by the reference phase, correcting any phase drifts. This adjusted signal is demodulated in a digital homodyne detection, integrated with a discrimination-optimizing window, and thresholded. 

One of the four-wave mixing tones is generated as the qubit and cavity tones by single-sideband modulation.  The other is digitally gated using the Agilent MXGs built-in pulse gating.  

\begin{figure}[H]
\centering
\label{fig:wiring}
\includegraphics[width=5.5in]{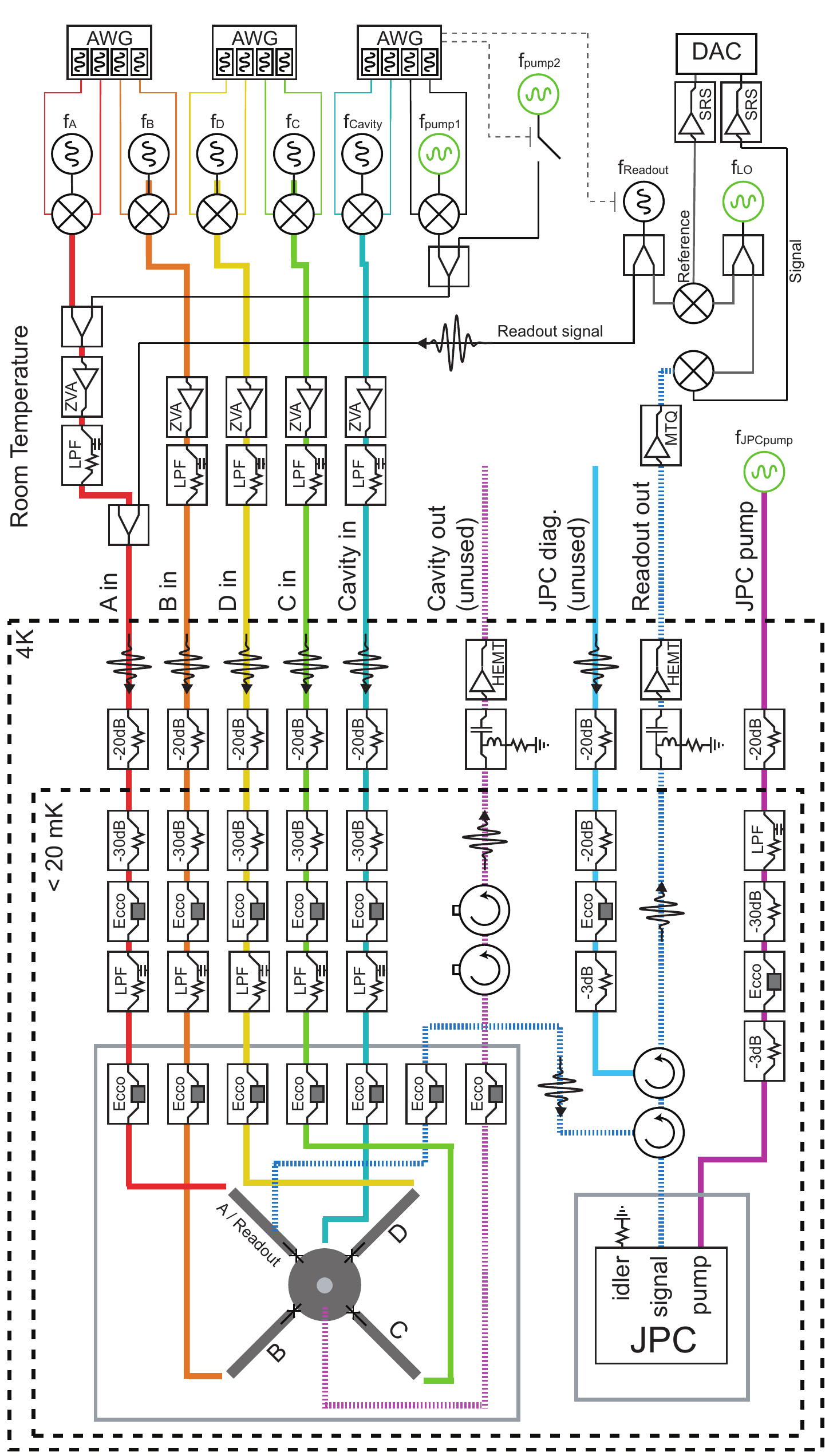}
\caption{Schematic of the dilution refrigerator wiring and experimental electronics}
\end{figure}


\section{Modifications on the original Nigg and Girvin proposal.}

The original Nigg and Girvin proposal suggested a particular hierarchy of qubit coupling strengths to the high-Q cavity: the ancilla coupling would be much greater than the sum of the register qubits couplings: $\chi_A \gg \sum_i \chi_i$. In this way, the original scheme can perform a gate selective on the ancilla state in a time that would result in negligible register-state induced cavity evolution. We have removed this constraint by applying register qubit echo sequences during long selective ancilla gates to cancel out the conditional cavity state evolution.

We have modified the cavity erasure step from the original proposal, which suggested the use of an ancilla-state dependent cavity displacement and required the coupling hierarchy described above. As discussed in the main text, after the ancilla mapping we recenter the cavity pointer states in phase space and invert the unitary evolution of the cavity mapping step to unconditionally erase the cavity and reset it to the vacuum. We remain agnostic to the state of the ancilla during this erasure by performing a full echo on the ancilla during the evolution.

The final modification from the Nigg-Girvin proposal is a method to speed up the gate for subset parity measurements that involve one and two qubits within a larger register, e.g. $ZII$ or $ZIZ$. 
In the main text we describe the cavity state as a pointer for the number of qubit excitations within the measured subspace, and each subspace will be entangled with cavity states differing in phase by $\theta$.
This mapping angle $\theta$ also sets the length of the mapping: $t_\text{map} = \theta/(2\pi\chi_\text{min})$.
In the original proposal, it was always assumed that $\theta = \pi$, leading to two cavity pointer states, one for each measurement outcome.

However, because we perform a binary mapping of the cavity pointer onto the ancilla, we really only require one result (even or odd parity) to be encoded into a single indistinguishable cavity pointer (i.e. the cavity states are the same).
For single-qubit measurements, we will by default have two such pointer states for any $\theta$. 
For two-qubit measurements, we will have in general three pointer states with one pointer state indicating the one-excitation or odd parity manifold, and two distinguishable cavity states for $\ket{gg}$ and $\ket{ee}$.
For $\theta = \pi$ the even-parity states also become indistinguishable, but this is unnecessary as we only require one indistinguishable parity pointer state--the odd-parity manifold. 
Accordingly, for these one- and two-qubit measurements, after the mapping phase angle is achieved, we displace the pointer associated with the odd-parity manifold to vacuum in order to entangle the ancilla qubit with the cavity state.
We note that this modification does not affect the ability to perform the erasure step.
Therefore, as long as we have sufficient separation (small wave-function overlap) between the cavity pointer states, we have the freedom to vary $\theta$ for mapping one- and two-qubit subset measurements.
By decreasing $\theta$, the mapping step takes less time at the cost of decreased separation between pointer states.
As a result, we add the initial displacement size as another degree of freedom to modify the timing of the algorithm.

The phase separation between cavity pointer states scales as $\theta = \chi_{min} t$, where $t$ is the length of the mapping.
The overlap between the two pointer states $\ket{\alpha}$ and $\ket{\beta}$ scales as $|\innerproduct{\alpha}{\beta}|^2 \approx e^{-\Delta}$ where $\Delta = 2n_0\left(1-\cos{\chi_{min}t}\right)$.
$\Delta$ can be understood as the distance between the two pointer states in units of photons if one state is at the vacuum, or as the square of the displacement between them in phase space.
We may then reduce the time for mapping by employing a larger initial displacement. For the one- and two-qubit measurements performed in the main text, we use $n_0 = 5$ and $\theta = 2\pi/5$, which leads to an overlap of $\sim 10^{-3}$. 
We restrict our initial displacement to these photon numbers to limit the dephasing effect of cavity self-Kerr, where different photon numbers $n$ acquire phase $\phi_n = n^2 Kt$.  
In general, we can use this modification to optimize the mapping protocol by adjusting time and initial photon population.

This speedup is not directly applicable to measurements of 3 or more qubits, and therefore in this work, we use $\theta=\pi$ for the $ZZZ$ measurement.
Though not explored, it is possible to apply this speedup to larger operators by decomposing them into a series of one- and two-qubit measurements, but without measuring or resetting the ancilla between them. 
This approach allows for the speedup method as discussed, but comes at a cost of scaling in time with the number of measured qubits $N$ as $\sim \mathrm{ceil}\left(N/2\right)$.

\section{Error budget from simulations}

Using simulations in QuTip \cite{qutip}, we have a rough estimate of our error budget for our process fidelities.  As the quantum instrument S-fidelity is not calculated quickly, these simulations were evaluated using the more conventional J-fidelity definition. Using this definition our experimental, unheralded measurements all yielded process fidelities of 0.77-0.83, and heralded datasets yielded 0.85-0.87.

From a lossless simulation with a undriven Hamiltonian of only dispersive interactions and cavity Kerr, we see roughly $5\%$  infidelity.  These we deem ``control and Hamiltonian errors,'' which include inaccuracies in delays, displacement, phases, non-orthogonality of coherent states, imperfect selectivity of our ancilla entangling pulse, and the difficulty of doing a perfectly non-selective $X_{\pi}$ gate with photons present.   We found that including direct qubit-qubit interactions has a negligible effect on the fidelity.  The number-state dependent dispersive shift has not been simulated, but likely has an effect on the same order as Kerr, however with an opposite sign.  This implies that in some situations the $\chi'$ term will actually combat the effect of cavity Kerr, albeit in a state-dependent manner.  Looking forward, we are planning to remedy these error sources using optimal control techniques.  

Including cavity loss and qubit $T_1$ and $T_2$ decreases the fidelity a further $6\%$.  We use the Ramsey oscillation time constants for $T_2$, except for qubit D as it does not exhibit clean Ramsey fringes.  For qubit D we use the Hahn echo time constant.   Roughly half of this infidelity comes from cavity photon loss.

Measurement infidelity of the ancilla itself costs a further $2\%$.  Transmon decoherence during the measurement (and during a further 300ns delay which was present only for technical reasons), costs another $2\%$.   Collectively these result in an $85\%$ fidelity, slightly above the observed $\sim 80\%$.  We ascribe the additional loss to apparent infidelity due to photon-dependence in our tomography (discussed in the state tomography portion of the supplement), along with uncertainty in treating the dephasing of qubit D and the presence of $\chi'$.

Projecting onto zero photons in simulation predicts a fidelity of $90\%$.   This  process exposes the system to an extra 1.8$\mu s$ of decoherence,  though we do apply spin-echo during this period.  Accordingly we predict an additional fidelity decay to $85\%$, in good agreement with our experimental results for heralded data which average $86\%$.   Note that as  the heralded data has zero residual photons, it does not suffer from the same potential state-tomography effects as the unheralded data.

\section{Post-Selective Cooling}

\begin{figure}[h]
\centering
\includegraphics[width=3.5in]{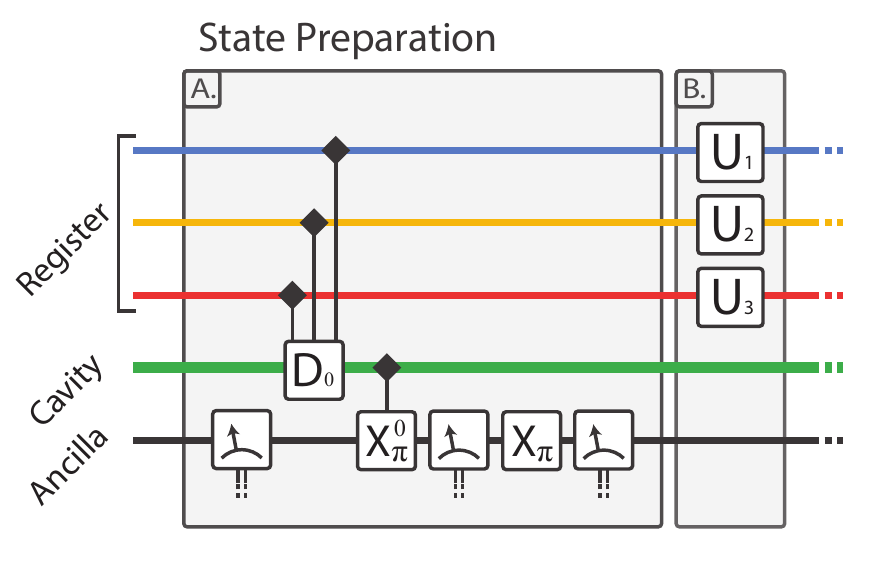}
\caption{\label{fig:state_prep}{
State preparation with post-selection. We prepare the ground state post-selectively before rotating the register into an initial state.  The first measurement confirms the ancilla is in the ground state.  We then selectively displace the cavity if the register is not in the ground state, then excite the ancilla if and only if the cavity is still in the ground state.   We then measure to confirm that the ancilla is excited.   Finally we unselectively invert the ancilla once more and check that it is in the ground state.  The measurement signature of ground, excited, ground, leads to a cold system with ground state probability of roughly $99\%$.
}}
\end{figure}

Our experimental procedure begins by performing a series of measurements for post-selective ground state preparation (depicted in Fig. \ref{fig:state_prep}, A).  We first measure the ancilla and keep the instances where we find the ancilla in the ground state.  We then displace the cavity dependent on any of the register qubits being excited, using a superposition of seven spectrally narrow displacement pulses.  We use a spectrally narrow $X_\pi$ pulse to excite the ancilla if and only if there are 0 photons, measure the ancilla, and condition on an excited state result.  We then perform an unselective $X$ gate on the ancilla followed by a measurement and condition on a ground state result.  This protocol leads to a cold system.  In this way we postselect away roughly $12\%$ of the data, consistent with independently measured background excitation probabilities.  After ground-state preparation we rotate the register qubits into a desired initial state.

Using the RPM protocol \cite{kurtis_rpm} we find initial excited-state populations of the qubits (in descending resonance frequency, starting with the ancilla), of $1.4$, $2.4$, $2.5$ and $3.5\%$.  After employing the state-preparation procedure in the main text, we find very low excited-state probabilities that are difficult to discern with $100,000$ shots.  Using the contrast of the signal, we bound the excitation probabilities to be below 0.05, 0.4, 0.08, and 0.2$\%$.  The initial cavity excitation probability is $3\%$, and after post-selection we bound its excitation probability to $1\%$.

\section{State and process tomography}

\begin{figure}[h]
\centering
\includegraphics[width=3.5in]{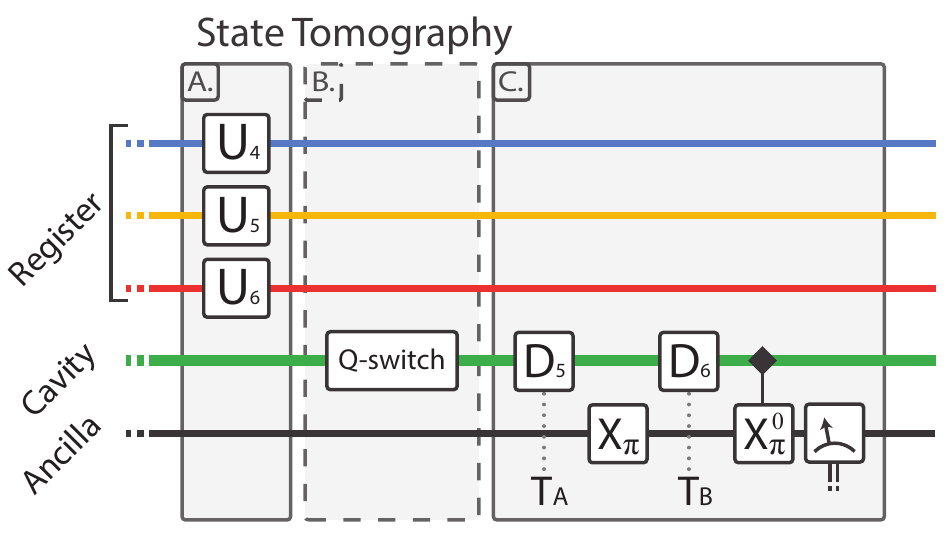}
\caption{\label{fig:tomography}{
Our tomographic step after the experiment is in three parts.  First, we rotate into the measurement basis.  Second, if we have not heralded on an empty cavity, we apply pump tones to empty the cavity, as discussed elsewhere in the supplement.  The third step flips the ancilla if and only if the register is in the ground state.
}}
\end{figure}

State tomography in this system is complicated by the fact that we have not installed individual readout lines for each register qubit.  Accordingly tomography of the register state must also be mediated by the cavity and readout via the ancilla.  The protocol we employ uses a toolbox similar to the more complicated measurements under examination, but is simpler as we do not require it to be non-demolition. After the ancilla measurement, the tomographic mapping consists of an unselective cavity displacement, a delay of 136 ns, an unselective $X_{\pi}$ gate on the ancilla, a second delay of 136 ns, and a second unconditional displacement.  The displacement phases are chosen such that the coherent states entangled with $\ket{ggg}_r \otimes \ket{g}_a$ and $\ket{ggg}_r \otimes \ket{e}_a$ are entangled with the cavity being in the vacuum state, and all other register states sufficiently displaced from the vacuum.  We then perform a selective  $X_{\pi}^0$ gate on the ancilla, and measure again.  The signature of the ancilla \emph{changing} state then is an indicator of the register being in the vacuum.  The POVM enacted is effectively equivalent to the quantum bus measurement of $\ket{ggg}\bra{ggg}_r$ in \cite{qbus}.  

We calibrate the tomographic POVM $E_\text{tomo}$ assuming that we are only sensitive to the $Z$ projection of the register and preparing the eight computational states, then performing tomography.  We actually calibrate two POVMs, one assuming that the ancilla begins in the ground state and the other assuming the excited state.   We later showed this to be a good approximation by performing full quantum detector tomography on the tomographic measurement, with one of the POVMs depicted in Figure \ref{fig:tomo_povm}.

\begin{figure*}[h]
\centering
\includegraphics[width=7in]{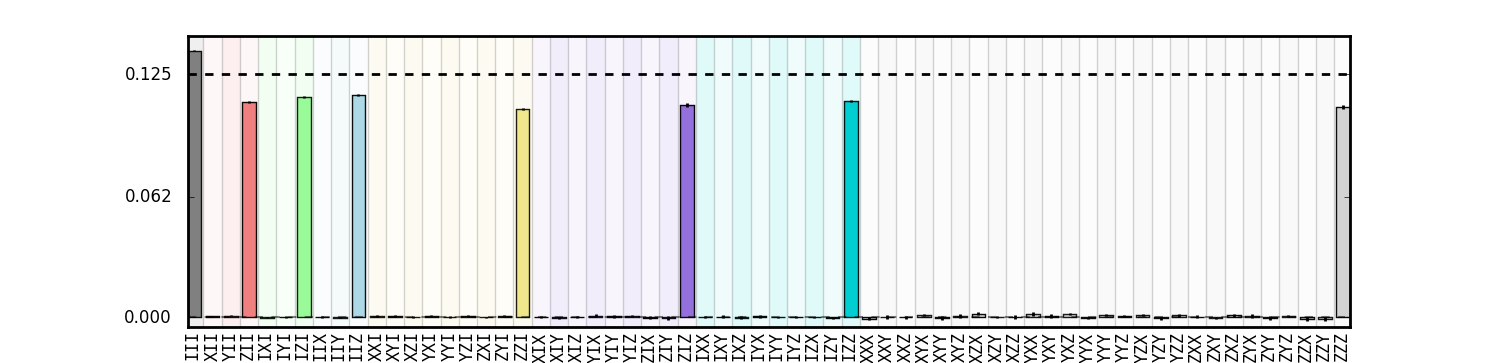}
\caption{\label{fig:tomo_povm}{
POVM of the tomography operator.
This depicts the POVM $E_\text{tomo}$ describing the tomography operator when the ancilla is initially in $\ket{g}$.  The $\ket{ggg}\bra{ggg}_r$ operator should ideally correspond to eight bars with amplitude 0.125. These results indicate the meter is only slightly biased toward the result when the ancilla measured to be in $\ket{g}$ and we have high contrast of the projector onto the register ground state. From these results we also conclude that the tomography measurement is not sensitive to Pauli operators with single qubit $X$ or $Y$ Pauli components, indicating that our calibration over only computational states is sufficient.
}}
\end{figure*}

Using this calibration, we perform tomography using an overcomplete set of post-rotations $\left\lbrace I, R_y\left(\pi/2\right), R_x\left(\pi/2\right), R_y\left(-\pi/2\right), R_x\left(-\pi/2\right), R_y\left(\pi\right) \right\rbrace ^{\otimes 3}$ on the conditioned output state.  We take $2000$-$5000$ averages per measurement. We compose a (non-square) matrix that relates the calibration data and the employed post-rotations to the expectation values of the generalized Pauli operators.  The elements of this matrix are derived from the measurement results $\lbrace m_i \rbrace = \text{Tr} \left[ E_\text{tomo}^l R_i \rho R_i^{\dag} \right] $, where $E_\text{tomo}$ is the tomographic POVM and the superscript $l$ indicates the initial state of the ancilla, conditioned on the prior measurement.   We expand $\rho$ as $\Sigma_j c_j \sigma_j$ where $\sigma_j$ are tensored Pauli operators.   This leads to the matrix $A^l_{ij} = \text{Tr} \left[ E_\text{tomo}^l R_i \sigma_j R_i^{\dag} \right]$. We can use the pseudo-inverse to perform an unconstrained least-squares inversion, but the data presented here result from a maximum-likelihood fit constraining normality and positivity of the reconstructed states. This convex optimization was performing using the CVXPY library \cite{cvxpy}.  

Process tomography is performed by injecting a complete set of initial states, using pre-rotations $\left\lbrace I, R_y\left(\pi/2\right), R_x\left(\pi/2\right), R_y\left(-\pi/2\right) \right\rbrace ^{\otimes 3}$.  The two (outcome-dependent) density matrices, which are normalized in the state tomography process, are multiplied by their respective probabilities to have occurred. 
We use these two pairs of input-output density matrices to perform a complete reconstruction of the quantum instrument superoperator. This procedure is implemented as a semi-definite program using the \textsc{CVX} package in \textsc{MATLAB} to perform least-squared fit with constraints for a positive and trace-preserving process.

One notable complication is that this mapping assumes that there are zero photons in the cavity, or at least appreciably close to zero and that the photon-number distribution is not state-dependent.  This is not generally true following our procedure when we do not herald on there being zero photons.  To unconditionally empty the cavity in this case we perform a four-wave mixing procedure, detailed later in this supplement.  The qubit decay during this period is calibrated into the $E_\text{tomo}^l$ matrices.

We have also examined the effects of these residual photons on the tomography via simulation.  Intuitively, it would take pathological behavior of the photons to over-estimate the state fidelity in this experiment, as the general result is to report measurements that should indicate "NOT $\ket{ggg}$" as "YES $\ket{ggg}$". Averaging over a full set of post-rotations this would serve to decrease contrast in all cases.  We have done time-domain simulations of the tomographic procedure with QuTip \cite{qutip}.  Using actual output states of our simulated algorithm, we have simulated fidelity as-is and after manual tracing over and removing all residual photons.  In all cases the fidelity has been higher after removal of photons, suggesting if there are residual photons affecting our tomography, the experimentally reported results are lower bounds on the state fidelity. 

\newpage
\section{Process tomography:  Full results}
\begin{figure}[H]
\centering
\includegraphics[width=7in]{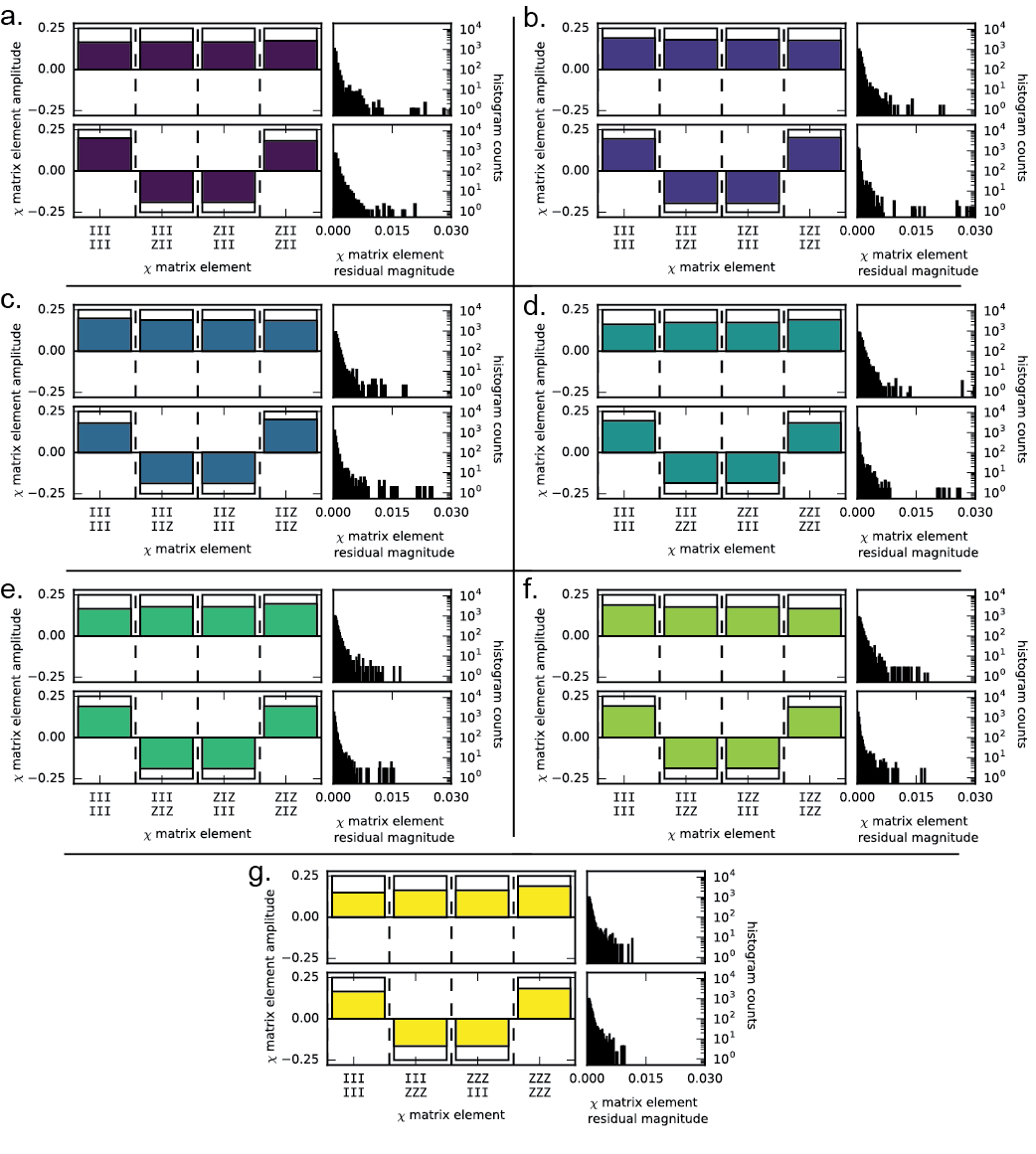}
\caption{\label{fig:all_qpt_u} Full conditioned QPT results. (a)-(g) show conditioned QPT results for all seven subset-parity operators: $ZII$, $IZI$, $IIZ$, $ZZI$, $ZIZ$, $IZZ$, and $ZZZ$, respectively. For each panel, we represent the results in the $\chi$ matrix representation for both even (top row) and odd (bottom row) outcomes. For each outcome, the left bar plot directly shows the four non-zero components of the $\chi$ matrix with the ideal amplitudes outlined with amplitude $\pm 1/4$. Note the sign-change for the off-diagonal components between even and odd outcomes for all operators. The right plot is a histogram that illustrates the magnitude for the remaining (ideally zero) components of the $\chi$ matrix. We note that they are all small and with few outliers.}
\end{figure}

\begin{figure}[H]
\centering
\includegraphics[width=7in]{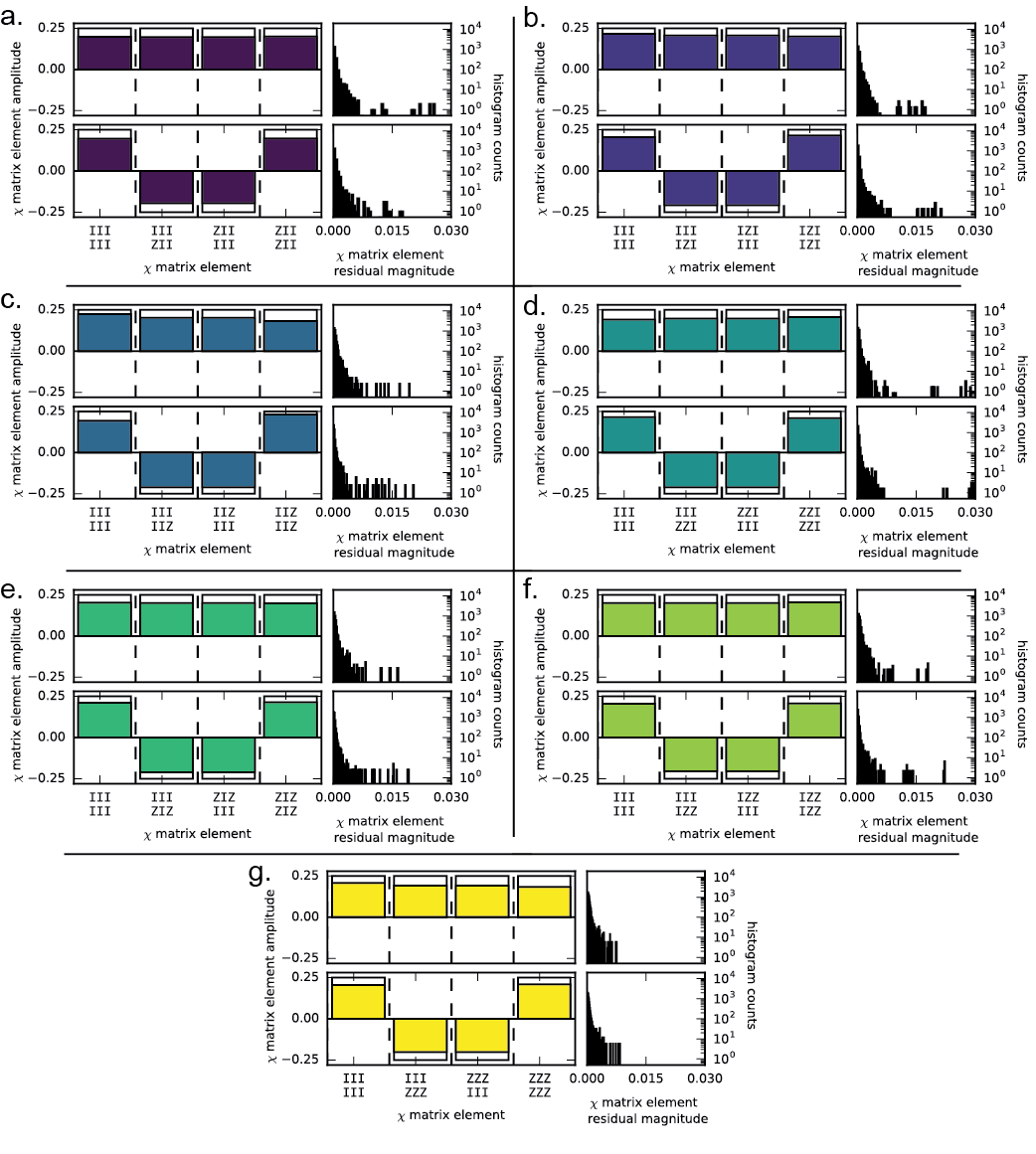}
\caption{\label{fig:all_qpt_h} Conditioned QPT results with success herald. (a)-(g) show conditioned QPT results for all seven subset-parity operators: $ZII$, $IZI$, $IIZ$, $ZZI$, $ZIZ$, $IZZ$, and $ZZZ$, respectively. The representation of this data is the same as in Fig. \ref{fig:all_qpt_u}.}
\end{figure}

\newpage
\section{Detector Tomography}
We can experimentally determine the POVM operators in a manner akin to state tomography.  We can see this by looking at the symmetry between $\rho$ and $E$ in the probability of a POVM outcome $\pi = \Tr\left[E \rho \right]$. In state tomography we assume that we know $E$. We reconstruct an unknown state $\rho$ by applying a set of post-rotations and then measure the resulting states with $E$. In POVM tomography we assume that we do not know $E$, but we do perfectly know $\rho$, e.g. the ground state. The operation is then symmetric. We describe $E$ in the Pauli basis and choose $\rho = \outerproduct{ggg}{ggg}$, the register ground state.

\begin{align}
 m_i  &= \Tr \left[ E R_i \outerproduct{ggg}{ggg} R_i^{\dag} \right]\\
 m_i  &= \Tr \left[ \sum_j c_j \sigma_j R_i \outerproduct{ggg}{ggg} R_i^{\dag} \right]\\
 m_i  &= \sum_j c_j \Tr \left[ \sigma_j R_i \outerproduct{ggg}{ggg} R_o^{\dag} \right]
\end{align}

And we have the same matrix inversion problem as in state tomography. 

We have performed detector tomography on all seven measurement operators, both without heralding on zero photons and with the success-heralding measurement. The full results from reconstructions of our detector POVM are shown in Figs. \ref{fig:povm_data_unheralded} and \ref{fig:povm_data_heralded}, respectively.  
Full results are discussed in the next section.

\newpage
\section{QDT: Full results}

\begin{figure}[H]
\label{fig:povm_data_unheralded}
\includegraphics[width=7in]{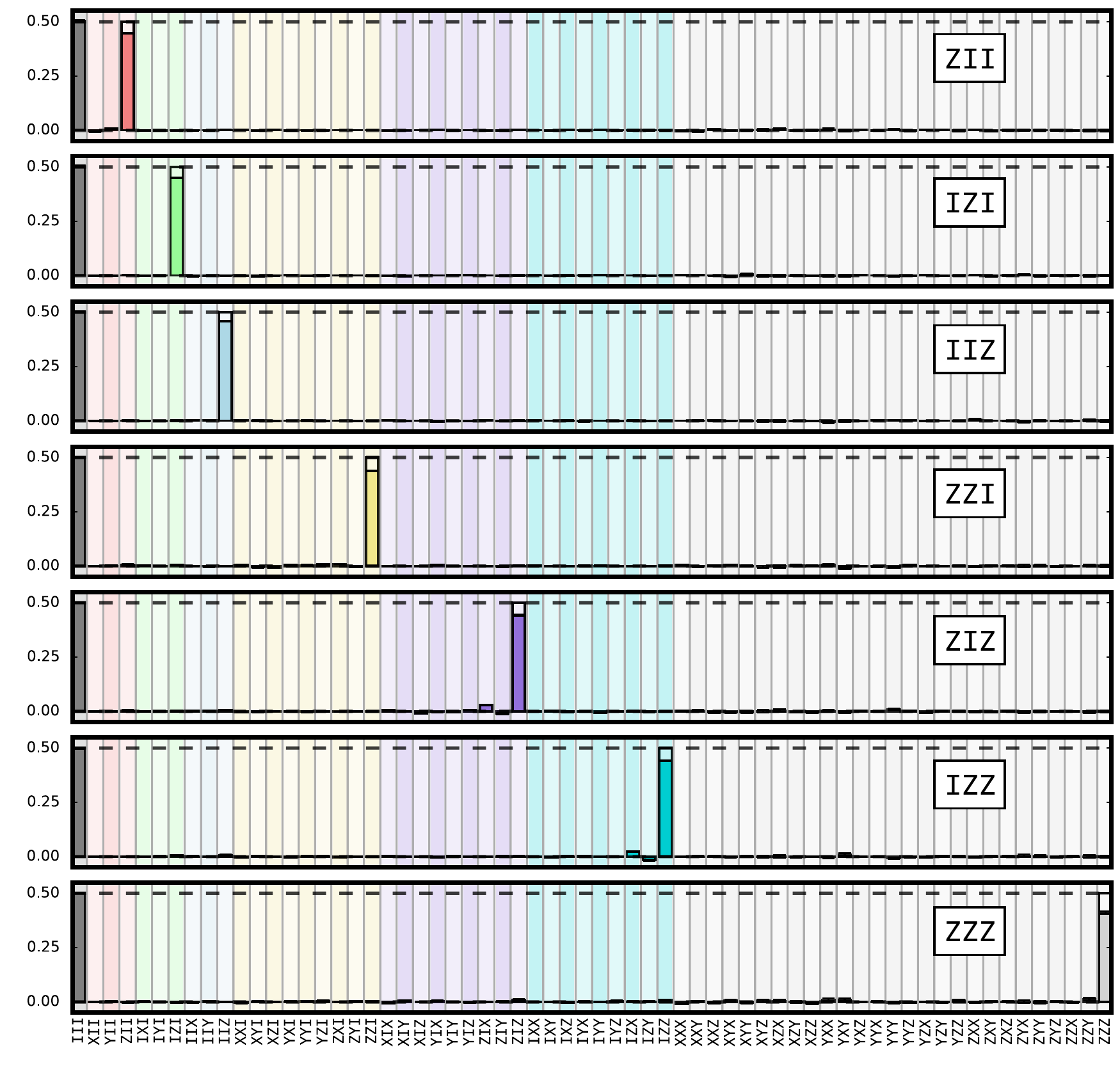}
\caption{Tomography of all POVMs, without post-selection on the success-herald}
\end{figure}

\begin{figure}[H]
\label{fig:povm_data_heralded}
\includegraphics[width=7in]{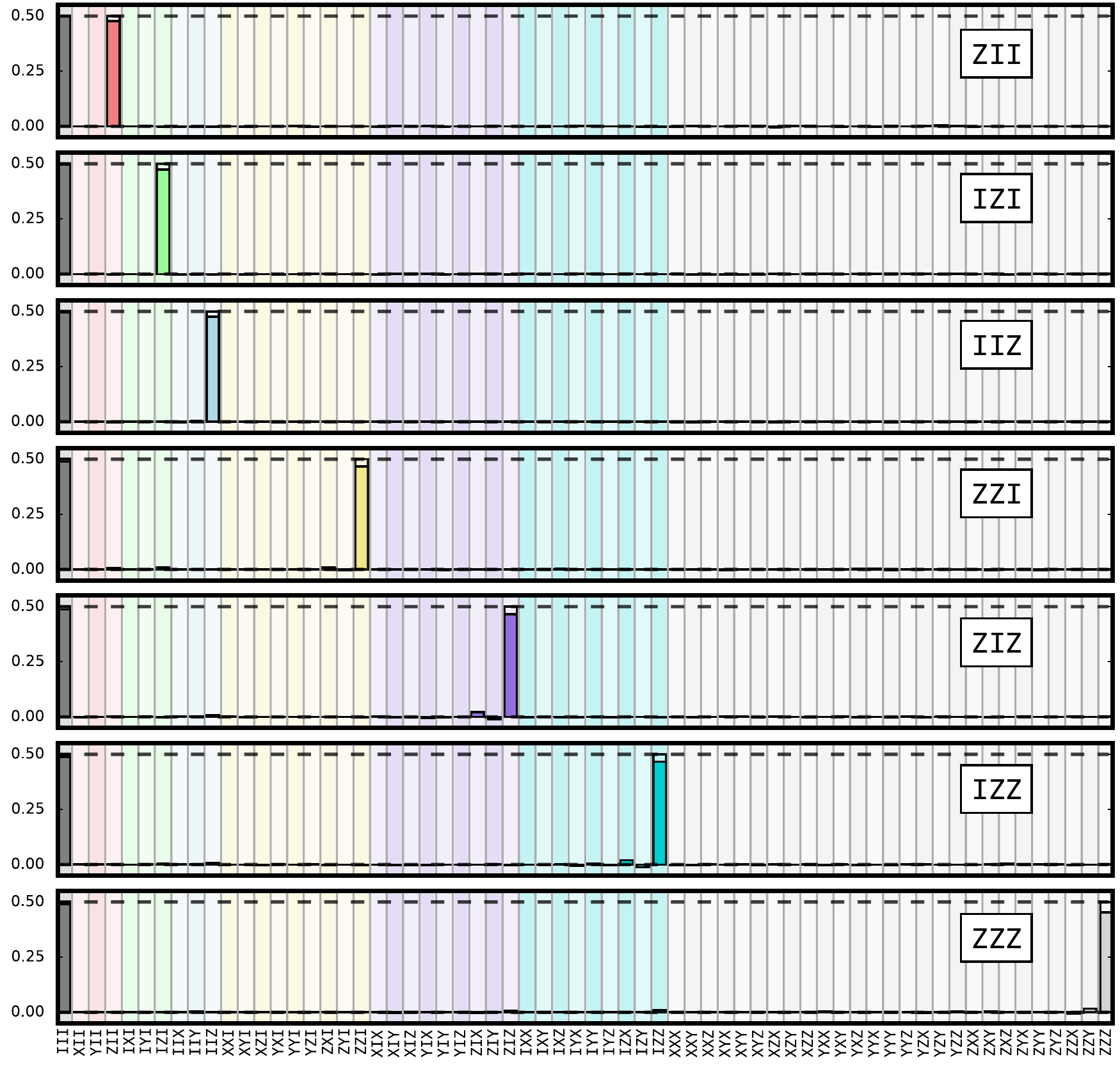}
\caption{Tomography of all POVMs, with post-selection on the success-herald}
\end{figure}

\section{Figures of merit: Full results}

We present an exhaustive list of the figures of merit used in this work.
\begin{itemize}
\item Table \ref{table:assi_fids} details the assignment fidelities for each measurement operator.
\item Table \ref{table:thetas_results} details the specificity for each measurement operator.
\item Table \ref{table:measure_def} details the formulas used for the following results, including the distance metrics.
\item Table \ref{table:povm_results} shows the J- and S-measures for the detector.
\item Table \ref{table:qi_results} shows the J- and S- measures for the measurement process.
\end{itemize}

\begin{table}[H]
\caption{\label{table:assi_fids} Assignment fidelity results}
\begin{ruledtabular}
\begin{tabular}{ c c c } 

operator & unheralded results & heralded results \\
\hline
ZII	& 0.93 & 0.97 \\
IZI	& 0.94 & 0.97 \\
IIZ	& 0.94 & 0.96 \\
\hline
ZZI	& 0.93 & 0.97 \\
ZIZ	& 0.93 & 0.96 \\
IZZ & 0.93 & 0.97 \\
\hline
ZZZ	& 0.89 & 0.94 \\
\end{tabular}
\end{ruledtabular}
\end{table}

\begin{table}[H]
\caption{\label{table:thetas_results} Specificity results.}
\begin{ruledtabular}
\begin{tabular}{ c c c } 

operator & unheralded results & heralded results \\
\hline
ZII	& 2.3$^{\circ}$ & 1.1$^{\circ}$ \\
IZI	& 1.8$^{\circ}$ & 1.0$^{\circ}$  \\
IIZ	& 1.8$^{\circ}$ & 0.9$^{\circ}$ \\
\hline
ZZI	& 3.6$^{\circ}$ & 2.3$^{\circ}$ \\
ZIZ	& 5.2$^{\circ}$ & 3.4$^{\circ}$ \\
IZZ & 4.5$^{\circ}$ & 3.4$^{\circ}$ \\
\hline
ZZZ	& 4.4$^{\circ}$ & 3.0$^{\circ}$ \\
\end{tabular}
\end{ruledtabular}
\end{table}

\begin{table}[H]
\caption{\label{table:measure_def} List of detector and process measures. We define $J_\mathcal{E}$ as the Jamio\l{}kowski matrix representing the process $\mathcal{E}$.}
\begin{ruledtabular}
\begin{tabular}{ c | c c } 
definition & fidelity & distance \\
\hline
state & 
$\mathcal{F}_\text{state} \left( \rho, \sigma \right) = \Tr \left( \sqrt{\rho^{1/2} \sigma \rho^{1/2}} \right)^2$
& 
$\mathcal{D}_\text{state} \left( \rho, \sigma \right) = \frac{1}{2} ||\rho - \sigma||_1$
\\
Jamio\l{}kowski (process) &
$\mathcal{F}_{\text{J}} \left(  \mathcal{E}, \mathcal{F} \right) = 
\mathcal{F}_{\text{J}} \left(  J_\mathcal{E}, J_\mathcal{F} \right)$
& 
$\mathcal{D}_\text{pro} \left(  \mathcal{E}_1, \mathcal{E}_2 \right) = 
\mathcal{D}_\text{state} \left(  J_\mathcal{E}, J_\mathcal{F}  \right)$
\\
Stabilized (worst case) & 
$\mathcal{F}_{\text{S}} \left( \mathcal{E}_1, \mathcal{E}_2 \right) = $
& 
$\mathcal{D}_{\text{S}} \left( \mathcal{E}_1, \mathcal{E}_2 \right) = $
\\
&
$ \min_{\rho_{SA}} \mathcal{F}_{\text{state}} \Big( \left(\mathcal{E}_1 \otimes I \right) \rho_{SA},  \left(\mathcal{E}_2 \otimes I\right) \rho_{SA} \Big)$
& 
$\max_{\rho_{SA}} \mathcal{D}_{\text{state}} \Big( \left(\mathcal{E}_1 \otimes I\right) \rho_{SA},  \left(\mathcal{E}_2 \otimes I\right) \rho_{SA} \Big)$
\\
\end{tabular}
\end{ruledtabular}
\end{table}

\begin{table}[H]
\caption{\label{table:povm_results} POVM results.}
\begin{ruledtabular}
\begin{tabular}{ c c c c c c c c c } 
 & \multicolumn{4}{c}{unheralded results} & \multicolumn{4}{c}{heralded results} \\
\cline{2-5} \cline{6-9}
operator & $F_{\text{J}}$ & $D_{\text{J}}$ & $F_{\text{S}}$ & $D_{\text{S}}$
		 & $F_{\text{J}}$ & $D_{\text{J}}$ & $F_{\text{S}}$ & $D_{\text{S}}$ \\
\hline
ZII	& 
	0.948 & 0.055 & 0.936 & 0.068 &
    0.976 & 0.025 & 0.967 & 0.032 \\
IZI	& 
	0.945 & 0.053 & 0.931 & 0.071 &
    0.975 & 0.026 & 0.963 & 0.038  \\
IIZ	& 
	0.958 & 0.044 & 0.948 & 0.054 &
    0.976 & 0.025 & 0.967 & 0.035 \\
\hline
ZZI	& 
	0.938 & 0.065 & 0.906 & 0.096 &
    0.967 & 0.036 & 0.938 & 0.065 \\
ZIZ	& 
	0.943 & 0.068 & 0.914 & 0.094 &
    0.967 & 0.043 & 0.938 & 0.068 \\
IZZ	& 
	 0.942 & 0.066 & 0.918 & 0.087 &
     0.966 & 0.043 & 0.934 & 0.073 \\
\hline
ZZZ	& 
	0.910 & 0.093 & 0.875 & 0.130 &
    0.954 & 0.050 & 0.926 & 0.078 \\
\end{tabular}
\end{ruledtabular}
\end{table}

\begin{table}[H]
\caption{\label{table:qi_results} Quantum instrument metrics.}
\begin{ruledtabular}
\begin{tabular}{ c c c c c c c c c } 
 & \multicolumn{4}{c}{unheralded results} & \multicolumn{4}{c}{heralded results} \\
\cline{2-5} \cline{6-9}
operator & $F_{\text{J}}$ & $D_{\text{J}}$ & $F_{\text{S}}$ & $D_{\text{S}}$
		 & $F_{\text{J}}$ & $D_{\text{J}}$ & $F_{\text{S}}$ & $D_{\text{S}}$ \\
\hline
ZII	& 
	0.716  & 0.373 & 0.566 & 0.501 &
    0.789  & 0.312 & 0.690 & 0.400 \\
IZI	& 
	0.757 & 0.295 & 0.615 & 0.428 &
    0.832 & 0.246 & 0.749 & 0.345 \\
IIZ	& 
	0.757 & 0.223 & 0.607 & 0.278 &
    0.827 & 0.303 & 0.756 & 0.417 \\
\hline
ZZI	& 
	0.721 & 0.264 & 0.629 & 0.315 &
    0.818 & 0.333 & 0.754 & 0.411 \\
ZIZ	& 
	0.735 & 0.297 & 0.638 & 0.388 &
    0.823 & 0.220 & 0.757  & 0.272 \\
IZZ	& 
	0.730 & 0.306 & 0.595 & 0.437 &
    0.813 & 0.238 & 0.759 & 0.272 \\
\hline
ZZZ	& 
	0.674 & 0.353 & 0.578 & 0.450 &
    0.796 & 0.221 & 0.741 & 0.272 \\
\end{tabular}
\end{ruledtabular}
\end{table}


\section{Inherent stability of two-outcome POVM fidelity}
We calculate the stable diamond norm of the difference of a pair of arbitrary 2-outcome POVM channels and construct the the optimal discrimination input state. We will show that optimal discrimination does not require an ancilla entangled with the system on which the POVMs act. This means for POVM channels the stabilized (including ancilla) and unstabilized
(no ancilla) diamond norm are equal. As a preparation we first define a few concepts and introduce a necessary lemma. 
\begin{defn}
\textbf{(POVM channel)} For a $m$-outcome, $d$-dimensional POVM,
$\{M_{\mu}\}$ where $M_{\mu}\succeq0$ and $\sum_{i}^{m}M_{\mu}=\mathbb{I}_{d\times d}$,
define the quantum-to-classical channel,
\[
\mathcal{E}(\rho)=\sum_{\mu}\Tr{(M_{\mu}\rho)}\ket{\mu}\bra{\mu}_d
\]
where $\ket{\mu}_{d}$ are the pointer states of the detector. Note that
$\mathcal{E}(\rho)$ has input dimension $d$ and output dimension
$m$. 
\end{defn}

\begin{defn}
\textbf{(Diamond norm)} For two channels $\mathcal{E}_{1}$ and $\mathcal{E}_{2}$
having the same input/output space, define the diamond norm
\[
\left\Vert \mathcal{E}_{1}-\mathcal{E}_{2}\right\Vert _{\diamond}\equiv\frac{1}{2}\max_{\rho_{SA}}\left\Vert \mathcal{E}_{1}\otimes\mathbb{I}(\rho_{SA})-\mathcal{E}_{2}\otimes\mathbb{I}(\rho_{SA})\right\Vert _{*},
\]
where $\left\Vert \cdot\right\Vert _{*}$is the trace norm (nuclear
norm), and S denotes the system $\mathcal{E}_{1}$ and $\mathcal{E}_{2}$
act on and A denotes an ancilla potentially entangled with S. Note
that due to convexity of both the trace norm and the space of density
matrices $\rho_{SA}$, the maximizer is always pure, $\rho_{SA}=\ket{\psi_{SA}}\bra{\psi_{SA}}$. 
\end{defn}

Similarly we can define a norm without the inclusion of the ancilla.
\begin{defn}
\textbf{(Unstable trace norm)} For two channels $\mathcal{E}_{1}$
and $\mathcal{E}_{2}$ having the same input/output space, define
the worst-case trace norm
\[
\left\Vert \mathcal{E}_{1}-\mathcal{E}_{2}\right\Vert _{*}\equiv\frac{1}{2}\max_{\rho_{S}}\left\Vert \mathcal{E}_{1}(\rho_{S})-\mathcal{E}_{2}(\rho_{S})\right\Vert _{*},
\]
Due to convexity of both the trace norm and the space of density matrices
$\rho_{S}$, the maximizer is always pure, $\rho_{S}=\ket{\psi_{S}}\bra{\psi_{S}}$. 
\end{defn}

\begin{lem}
\label{lem:nuclear-norm-maxim}Let $A$ and $B$ be arbitrary complex
$d\times d$ matrices and let $W$ be unitary with same dimension.
Using the singular value decomposition, $A=U_{1}S_{1}V_{1}^{\dagger}$,
$B=U_{2}S_{2}V_{2}^{\dagger}$, where $S_{1}$ / $S_{2}$ have decreasing
diagonal elements, then
\[
\left\Vert A^{\dagger}WB\right\Vert _{*}\le\sum_{i}\sigma_{i}(A)\sigma_{i}(B)
\]
where $\sigma_{i}(M)$ stands for the $i$-th largest singular value
of matrix $M$. The equality holds when $W=U_{1}U_{2}^{\dagger}$. \end{lem}
\begin{proof}
The trace norm is unitarily invariant, i.e. $\left\Vert M\right\Vert _{*}=\left\Vert U\cdot M\cdot\tilde{U}\right\Vert _{*}$.
So we have 
\begin{eqnarray*}
\left\Vert A^{\dagger}WB\right\Vert _{*} & = & \left\Vert V_{1}S_{1}U_{1}^{\dagger}WU_{2}S_{2}V_{2}^{\dagger}\right\Vert _{*}\\
 & = & \left\Vert S_{1}U_{1}^{\dagger}WU_{2}S_{2}\right\Vert _{*}\\
 & = & \left\Vert S_{1}\tilde{W}S_{2}\right\Vert _{*}\\
 & = & \sum_{i}\sigma_{i}(S_{1}\tilde{W}S_{2})\\
 & \le & \sum_{i}\sigma_{i}(S_{1})\sigma_{i}(\tilde{W}S_{2})\\
 & = & \sum_{i}\sigma_{i}(S_{1})\sigma_{i}(S_{2})\\
 & = & \sum_{i}\sigma_{i}(A)\sigma_{i}(B)
\end{eqnarray*}
where the important inequality used above is due to Theorem IV.2.5
of R. Bhatia's \textit{Matrix Analysis} \cite{matrix_analysis}. It is easy to check that
when $\tilde{W}=\text{diag}(e^{i\phi_{j}})$, i.e. $W=U_{1}\text{diag}(e^{i\phi_{j}})U_{2}^{\dagger}$. 
\end{proof}

Now we can prove our main result. 
\begin{thm}
\textbf{(Diamond norm of 2-outcome POVM channels) }For two POVM channels
$\mathcal{E}_{1}$ and $\mathcal{E}_{2}$ corresponding to $\{M_{1},\, M_{2}\}$
and $\{N_{1},\, N_{2}\}$, the diamond norm\textbf{ }
\[
\left\Vert \mathcal{E}_{1}-\mathcal{E}_{2}\right\Vert _{\diamond}\equiv\max\{\left|eig(M_{1}-N_{1})\right|\}.
\]
 The optimal input state is a product state $\ket{\psi_{SA}}=\ket{\phi_{S}}\otimes\ket{\psi_{A}}$,
where $\ket{\phi_{S}}$ is the eigenvector of $(M_{1}-N_{1})$ corresponding
to the eigenvalue with the largest absolute value. This is to say,
for maximal distinguishability, the ancilla is not required and that
\[
\left\Vert \mathcal{E}_{1}-\mathcal{E}_{2}\right\Vert _{\diamond}=\left\Vert \mathcal{E}_{1}-\mathcal{E}_{2}\right\Vert _{*}.
\]
\end{thm}
\begin{proof}
Let 
\[
\ket{\psi_{SA}}=\sum_{i}^{d}\sqrt{p_{i}}\ket i_{S}\ket{e_{i}}_{A},
\]
where $p_{i}\ge0$ and $\sum_{i}^{d}p_{i}=1$. We then have
\begin{eqnarray*}
 &  & (\mathcal{E}_{1}-\mathcal{E}_{2})\left(\ket{\psi_{SA}}\bra{\psi_{SA}}\right)\\
 & = & (\mathcal{E}_{1}-\mathcal{E}_{2})\left(\sum_{ij}\sqrt{p_{i}p_{j}}\ket i_{S}\bra j\otimes\ket{e_{i}}_{A}\bra{e_{j}}\right)\\
 & = & \sum_{ij}\sqrt{p_{i}p_{j}}(\mathcal{E}_{1}-\mathcal{E}_{2})\left(\ket i_{S}\bra j\right)\otimes\ket{e_{i}}_{A}\bra{e_{j}}\\
 & = & \sum_{ij}\sqrt{p_{i}p_{j}}\sum_{\mu}\Tr{\left[(M_{\mu}-N_{\mu})\ket i_{S}\bra j\right]}\ket{\mu}\bra{\mu}\otimes\ket{e_{i}}_{A}\bra{e_{j}}\\
 & = & \sum_{ij}\sqrt{p_{i}p_{j}}\sum_{\mu}\bra j(M_{\mu}-N_{\mu})\ket i\ket{\mu}\bra{\mu}\otimes\ket{e_{i}}_{A}\bra{e_{j}}\\
 & = & \left[\begin{array}{cc}
\sum_{ij}\sqrt{p_{i}p_{j}}\bra j(M_{1}-N_{1})\ket i\ket{e_{i}}_{A}\bra{e_{j}} & \mathbf{0}_{D\times D}\\
\mathbf{0}_{D\times D} & \sum_{ij}\sqrt{p_{i}p_{j}}\bra j(M_{2}-N_{2})\ket i\ket{e_{i}}_{A}\bra{e_{j}}
\end{array}\right].
\end{eqnarray*}
Since $M_{1}+M_{2}=\mathbb{I}_{d\times d}$ and $N_{1}+N_{2}=\mathbb{I}_{d\times d}$,
we have $M_{1}-N_{1}=-(M_{2}-N_{2})$ and
\[
\left\Vert M_{1}-N_{1}\right\Vert _{*}=\left\Vert M_{2}-N_{2}\right\Vert _{*}
\]
Therefore 
\begin{eqnarray*}
 &  & \left\Vert (\mathcal{E}_{1}-\mathcal{E}_{2})\left(\ket{\psi_{SA}}\bra{\psi_{SA}}\right)\right\Vert _{*}\\
 & = & 2\left\Vert \sum_{ij}\sqrt{p_{i}p_{j}}\bra j(M_{1}-N_{1})\ket i\ket{e_{i}}_{A}\bra{e_{j}}\right\Vert _{*}\\
 & = & 2\left\Vert \left(\begin{array}{cccc}
\sqrt{p_{1}}\\
 & \sqrt{p_{2}}\\
 &  & \ddots\\
 &  &  & \sqrt{p_{D}}
\end{array}\right)(M_{1}-N_{1})^{T}\left(\begin{array}{cccc}
\sqrt{p_{1}}\\
 & \sqrt{p_{2}}\\
 &  & \ddots\\
 &  &  & \sqrt{p_{D}}
\end{array}\right)\right\Vert _{*}\\
 & = & 2\left\Vert D(M_{1}-N_{1})^{T}D\right\Vert _{*}\\
 & = & 2\left\Vert DU\Lambda U^{\dagger}D\right\Vert _{*}
\end{eqnarray*}
where we defined $D\equiv diag(\sqrt{p_{1}},\,\sqrt{p_{2}},\:\cdots,\:\sqrt{p_{d}})$
and diagonalized $(M_{1}-N_{1})^{T}=U\Lambda U^{\dagger}$, $\Lambda=diag(\lambda_{1},\,\lambda_{2},\,\cdots,\,\lambda_{d})$
with $\left|\lambda_{1}\right|\ge\left|\lambda_{2}\right|\ge\cdots\ge\left|\lambda_{d}\right|$. 

Using Lemma \ref{lem:nuclear-norm-maxim} twice we have
\begin{eqnarray*}
\left\Vert DU\Lambda U^{\dagger}D\right\Vert _{*} & \le & \sum_{i}\sigma_{i}(D)\sigma_{i}(\Lambda)\sigma_{i}(D)\\
 & \le & \sum_{i}\tilde{p}_{i}\left|\lambda_{i}\right|,
\end{eqnarray*}
where $\tilde{p}_{i}$ denote $p_{i}$ in descending order and we
also used the fact that singular values are the absolute values of
eigenvalues $\sigma_{i}=\left|\lambda_{i}\right|$. The weight sum
$\sum_{i}\tilde{p}_{i}\left|\lambda_{i}\right|$ achieves maximum
when we put all the weight on the largest element $\left|\lambda_{1}\right|$,
i.e. $\tilde{p}_{1}=1$ and $\tilde{p}_{i\ne1}=0$. Therefore
\[
\left\Vert DU\Lambda U^{\dagger}D\right\Vert _{*}\le\left|\lambda_{1}\right|
\]
and 
\begin{eqnarray*}
\left\Vert \mathcal{E}_{1}-\mathcal{E}_{2}\right\Vert _{\diamond} & = & \frac{1}{2}\max_{\rho_{SA}}\left\Vert \mathcal{E}_{1}\otimes\mathbb{I}(\rho_{SA})-\mathcal{E}_{2}\otimes\mathbb{I}(\rho_{SA})\right\Vert _{*}\\
 & \le & \left|\lambda_{1}\right|.
\end{eqnarray*}
 It is simple to verify that the equality is indeed achievable when
we pick 
\[
\ket{\psi_{SA}}=\ket{\phi_{1}}_{S}\otimes\ket{\psi}_{A},
\]
 where $\ket{\phi_{1}}_{S}$ is the eigenvector of $(M_{1}-N_{1})$
corresponding to the eigenvalue with largest absolute value $\left|\lambda_{1}\right|$. 
\end{proof}


%

\end{document}